\def\draft{0}
\def\doubleblind{0}
\documentclass[11pt,letterpaper]{article}

\usepackage[margin=1in]{geometry}
\usepackage[utf8]{inputenc}
\usepackage{amsthm}
\usepackage{amsthm, amsmath, amssymb, mathtools}
\usepackage{bm,bbm}
\usepackage{graphicx}
\usepackage{color,xcolor}
\usepackage{paralist,enumitem}
\usepackage{mathrsfs}
\usepackage[normalem]{ulem}
\usepackage{tabularx}
\usepackage{tikz}
\usepackage{caption,comment}
\usepackage{algorithmicx}
\usepackage{algorithm}
\usepackage{algpseudocode}
\usepackage{lipsum}
\newcounter{algsubstate}
\renewcommand{\thealgsubstate}{\alph{algsubstate}}

\algnewcommand\algorithmicinput{\textbf{Input:}}
\algnewcommand\Input{\item[\algorithmicinput]}

\algnewcommand\algorithmicoutput{\textbf{Output:}}
\algnewcommand\Output{\item[\algorithmicoutput]}

\algnewcommand\algorithmicgoal{\textbf{Goal:}}
\algnewcommand\Goal{\item[\algorithmicgoal]}

\newcommand{\blind}[2]{{\ifnum\draft=1\color{purple}\fi \ifnum\doubleblind=1#2\fi\ifnum\doubleblind=0#1\fi\ifnum\doubleblind=2$\{$ #1 $\vert$ #2 $\}$\fi}}

\makeatletter
\newcommand{\algmargin}{\the\ALG@thistlm}
\algnewcommand{\parState}[1]{\State%
  \parbox[t]{\dimexpr\linewidth-\algmargin}{\strut #1\strut}}
\usepackage[colorlinks=true, allcolors=blue]{hyperref}
\usepackage[capitalize,noabbrev,nameinlink]{cleveref}

\newcommand{\vnote}[1]{\ifnum\draft=1\textcolor{orange}{[\textbf{Santhoshini:} #1]}\fi}
\newcommand{\Exp}{\mathbb{E}}
\newcommand{\Var}{\mathsf{Var}}
\newcommand{\cnt}{\mathsf{count}}
\newcommand{\estdeg}{\mathsf{est}\text{-}\mathsf{deg}}
\newcommand{\appdeg}{\mathsf{approx}\text{-}\mathsf{deg}}
\newcommand{\storedeg}{\mathsf{store}\text{-}\mathsf{deg}}
\usepackage[style=alphabetic, backend=biber, minalphanames=3, maxalphanames=4, maxbibnames=99, maxcitenames=4]{biblatex}

\usepackage{amsthm}
\usepackage{thmtools,thm-restate}

\numberwithin{equation}{section}
\declaretheoremstyle[bodyfont=\it,qed=\qedsymbol]{noproofstyle}

\declaretheorem[name=Observation,numbered=no]{observation*}

\declaretheorem[numberlike=equation]{theorem}

\declaretheorem[name=Theorem,numbered=no]{theorem*}

\declaretheorem[name=Lemma,numbered=no]{lemma*}

\declaretheorem[numberlike=equation]{corollary}
\declaretheorem[name=Corollary,numbered=no]{corollary*}

\declaretheorem[name=Proposition,numbered=no]{proposition*}

\declaretheorem[numberlike=equation]{claim}
\declaretheorem[name=Claim,numbered=no]{claim*}

\declaretheorem[name=Conjecture,numbered=no]{conjecture*}

\declaretheorem[name=Question,numbered=no]{question*}

\declaretheoremstyle[bodyfont=\it]{defstyle} 

\declaretheorem[numberlike=equation,style=defstyle]{definition}
\declaretheorem[unnumbered,name=Definition,style=defstyle]{definition*}

\declaretheorem[unnumbered,name=Example,style=defstyle]{example*}

\declaretheorem[unnumbered,name=Notation=defstyle]{notation*}

\declaretheorem[unnumbered,name=Construction,style=defstyle]{construction*}

\declaretheoremstyle[]{rmkstyle}

\crefname{claim}{Claim}{Claims}
\crefname{fact}{Fact}{Facts}
\crefname{corollary}{Corollary}{Corollaries}

\newcommand{\mcut}{\mathsf{Max}\text{-}\mathsf{CUT}}
\newcommand{\mdcut}{\mathsf{Max}\text{-}\mathsf{DICUT}}

\newcommand{\dcut}{\mathsf{DICUT}}

\newcommand{\Bern}{\mathsf{Bern}}

\newcommand{\nbrtype}[4]{\mathsf{nbhdtype}_{#1,#2}^{#3}(#4)}

\newcommand{\DbAllTypesDeg}[3]{\mathbf{Typ}_{#1}^{#2,#3}}

\newcommand{\Unif}[1]{\mathsf{Unif}(#1)}

\newcommand{\vecsigma}{\boldsymbol{\sigma}}

\newcommand{\EdgeDist}[3]{\mathsf{EdgeNbhdTypeDist}_{#1;#2}^{#3}}

\newcommand{\tvdist}[2]{\mathsf{tvd}(#1,#2)}

\newcommand{\Local}{\mathtt{Local}}

\makeatletter
\newenvironment{breakablealgorithm}
  {
   \begin{center}
     \refstepcounter{algorithm}
     \hrule height.8pt depth0pt \kern2pt
     \renewcommand{\caption}[2][\relax]{
       {\raggedright\textbf{\fname@algorithm~\thealgorithm} ##2\par}%
       \ifx\relax##1\relax 
         \addcontentsline{loa}{algorithm}{\protect\numberline{\thealgorithm}##2}%
       \else 
         \addcontentsline{loa}{algorithm}{\protect\numberline{\thealgorithm}##1}%
       \fi
       \kern2pt\hrule\kern2pt
     }
  }{
     \kern2pt\hrule\relax
   \end{center}
  }
\makeatother

\allowdisplaybreaks

    \title{Near-optimal streaming approximation for Max-DICUT in sublinear space using two passes}

\ifnum\doubleblind=0
\author{Santhoshini Velusamy\thanks{University of Waterloo, Ontario, Canada. Supported in part by NSF award CCF 2348475 when the author was affiliated with Toyota Technological Institute at Chicago. Email: \texttt{santhoshini.velusamy@uwaterloo.ca}}}
\fi
\ifnum\doubleblind=1
\author{Anonymous authors}
\fi

\date{}
\begin{document}
\maketitle

\begin{abstract}
The Max-DICUT problem has gained a lot of attention in the streaming setting in recent years, and has so far served as a canonical problem for designing algorithms for general constraint satisfaction problems (CSPs) in this setting. A seminal result of Kapralov and Krachun [STOC 2019] shows that it is impossible to beat $1/2$-approximation for Max-DICUT in sublinear space in the single-pass streaming setting, even on bounded-degree graphs. In a recent work, Saxena, Singer, Sudan, and Velusamy [SODA 2025] prove that the above lower bound is tight by giving a single-pass algorithm for bounded-degree graphs that achieves $(1/2-\epsilon)$-approximation in sublinear space, for every constant $\epsilon>0$. For arbitrary graphs of unbounded degree, they give an $O(1/\epsilon)$-pass $O(\log n)$ space algorithm. Their work left open the question of obtaining $1/2$-approximation for arbitrary graphs in the single-pass setting in sublinear space. We make progress towards this question and give a two-pass algorithm that achieves $(1/2-\epsilon)$-approximation in sublinear space, for every constant $\epsilon>0$.
\end{abstract}

\section{Introduction}
The maximum directed cut problem ($\mdcut$) is the analogue of the maximum cut ($\mcut$) problem for directed graphs. In particular, given a directed graph $G(V,E)$, the $\mdcut$ problem asks to find an ordered bipartition (a ``left'' set $L$ and a ``right'' set $R$)  of the vertex set $V$ that maximizes the number of edges that go from $L$ to $R$. In contrast, in $\mcut$, we measure all the edges that cross $L$ and $R$. Both $\mcut$ and $\mdcut$ are \textsf{NP Hard} and so, we study algorithms that approximately solve these problems. An $\alpha$-approximation algorithm is one that outputs a solution that has a value which is at least $\alpha$ fraction of the optimum.
Besides being fundamental problems of their own right, $\mcut$ and $\mdcut$ also belong to a large class of combinatorial optimization problems called the constraint satisfaction problems (CSPs) which capture several fundamental optimization problems in computer science. CSPs are well studied in theoretical computer science, owing in part to their structure which allows for finite characterizations of the infinite set of problems they contain, in several computational settings. Their study has led to the discovery of several important algorithmic and complexity-theoretic tools including SDP approximations, the PCP theorem, and \textsf{Unique Games}.

In many computational settings, $\mcut$ has successfully served as a canonical problem in the study of CSPs. For example, in the polynomial time setting, Goemans and Williamson \cite{DBLP:journals/jacm/GoemansW95} gave the first SDP-based approximation algorithm for $\mcut$. And on the lower bound side, conditioned on the \textsf{Unique Games} conjecture, Khot, Kindler, Mossell, and O'Donnell \cite{DBLP:journals/siamcomp/KhotKMO07} constructed tight integrality gap instances showing that the Goemans-Williamson SDP for $\mcut$ is optimal. In a breakthrough follow-up work, Raghavendra \cite{DBLP:conf/stoc/Raghavendra08} generalized both these works to all CSPs and gives an optimal SDP-based characterization for the approximability of every CSP in polynomial time. However, in the streaming setting, $\mcut$ does not capture all CSPs.

In the single-pass streaming model of computation, the constraints of a CSP instance (or edges of a graph in the case of $\mcut,\mdcut$) arrive in a stream. The streaming algorithm is allowed one pass over this stream and is required to compute an approximation to the optimum \emph{value} of the instance, using limited space. The amount of space available to the algorithm is typically parameterized by $n$, the number of variables in the instance (or vertices of the input graph in the case of $\mcut,\mdcut$). For every CSP, there exists a trivial sampling-based algorithm \cite{AG09,DBLP:conf/stoc/Trevisan09} that can achieve $(1-\epsilon)$-approximation using $\tilde{O}(n)$ space. Therefore, the interesting regime is when the algorithm is allowed only $o(n)$ space. Surprisingly, $\mcut$ has no non-trivial approximation in this regime! In particular, $1/2$-approximation is trivial (just output $1/2$ for every input graph) since any graph has a bipartition with at least half the edges crossing the bipartition.
A series of works \cite{DBLP:conf/soda/KapralovKS15,DBLP:conf/soda/KapralovKSV17,DBLP:conf/stoc/KapralovK19}, culminating with the breakthrough result of Kapralov and Krachun \cite{DBLP:conf/stoc/KapralovK19} show that for any constant $\epsilon>0$, any single-pass streaming algorithm that achieves $(1/2+\epsilon)$-approximation for $\mcut$ requires $\Omega(n)$ space. A follow-up work due to Chou, Golovnev, Sudan, Velingker, and Velusamy \cite{DBLP:conf/stoc/ChouGSVV22} proves optimal inapproximability results for multiple CSPs in $o(n)$ space, generalizing techniques from the $\mcut$ lower bound \cite{DBLP:conf/stoc/KapralovK19}. 

On the algorithmic side, the first non-trivial approximation algorithm was discovered by Guruswami, Velingker, and Velusamy \cite{DBLP:conf/approx/GuruswamiVV17} for $\mdcut$. While $\mdcut$ has a trivial $1/4$-approximation, they gave a single-pass algorithm that achieves $(2/5-\epsilon)$-approximation using only $O(\log n)$ space, for every constant $\epsilon>0$. Chou, Golovnev, and Velusamy \cite{DBLP:conf/focs/ChouGV20} later achieved $(4/9-\epsilon)$-approximation in $O(\log n)$ space. They also prove that the $4/9$ factor is optimal in $o(\sqrt{n})$ space, generalizing the techniques from the $\Omega(\sqrt{n})$ space lower bound of $\mcut$ due to Kapralov, Khanna, and Sudan \cite{DBLP:conf/soda/KapralovKS15}. They further generalized their results to prove a $O(\log n)$ vs $\Omega(\sqrt{n})$ space dichotomy theorem for every $\mathsf{Max}$-$2\mathsf{CSP}$. This served as the basis for their subsequent works with Sudan \cite{chou2022approximabilitybooleancspslinear, DBLP:journals/jacm/ChouGSV24} which prove a $O(\log^3 n)$ vs $\Omega(\sqrt{n})$ space dichotomy characterization for every CSP.\footnote{Their dichotomy characterization is for \emph{dynamic} streaming algorithms which allows deletions and can process arbitrarily long streams.} Saxena, Singer, Sudan, and Velusamy \cite{DBLP:conf/focs/SaxenaS0V23} show that the $\Omega(\sqrt{n})$ lower bound in \cite{DBLP:conf/focs/ChouGV20} is tight and give a $0.485$-approximation algorithm for $\mdcut$ in $\tilde{O}(\sqrt{n})$ space. In a later work \cite{DBLP:conf/soda/SaxenaS0V25}, the same authors also show that it is possible to achieve $(1/2-\epsilon)$-approximation for \emph{bounded-degree} (i.e., maximum degree bounded by a constant) graphs in $O(n^{1-\Omega(1)})$ space, for every $\epsilon>0$. A trivial reduction from $\mcut$ shows that the factor $1/2$ is also optimal for $\mdcut$ in $o(n)$ space.

In the multi-pass setting, a streaming algorithm is allowed to do multiple passes over the same stream of constraints. Assadi and Nagargoje \cite{DBLP:conf/stoc/AssadiN21} show that any $p$-pass algorithm for $\mcut$ achieving $(1-\epsilon)$-approximation requires at least $\Omega(n^{1-O(\epsilon p)})$ space. A recent breakthrough result of Fei, Minzer, and Wang \cite{DBLP:journals/corr/abs-2503-23404} shows that for any $\epsilon>0$, any $p$-pass algorithm for $\mcut$ (also $\mdcut$) achieving $(1/2+\epsilon)$-approximation requires at least $\Omega(n^{1/3}/p)$ space.  In contrast, for \emph{dense} instances of average degree at least $n^\delta$, a two-pass streaming algorithm due to Bhaskara, Dharuki, and Venkatasubramanian \cite{DBLP:conf/icalp/BhaskaraDV18}, achieves $(1-\epsilon)$-approximation for $\mcut$ (also $\mdcut)$ using $\tilde{O}(n^{1-\delta})$ space. An algorithm due to Saxena, Singer, Sudan, and Velusamy achieves $(1/2-\epsilon)$-approximation on arbitrary graphs in $O(1/\epsilon)$ passes using $O(\log n)$ space \cite{DBLP:conf/soda/SaxenaS0V25}. In a prior work \cite{DBLP:conf/soda/SaxenaS0V23}, the same authors give a two-pass algorithm that achieves $0.485$-approximation using $O(\log n)$ space. A recent result due to Fei, Minzer, and Wang \cite{fei2025dichotomytheoremmultipassstreaming} gives a dichotomy characterization in the multi-pass setting based on the basic LP relaxation for every CSP. In particular, they show that the techniques from \cite{DBLP:conf/stoc/Trevisan01,DBLP:conf/soda/SaxenaS0V25} can be combined with the distributed algorithms of Yoshida \cite{Yos11} to obtain a streaming algorithm that solves the LP relaxation for every CSP in $O(1/\epsilon^2)$ passes and $O(\log n)$ space. And to beat the approximation achieved  by the basic LP, any $p$-pass streaming algorithm requires at least $\Omega(n^{1/3}/p)$ space. 

The connection between streaming algorithms and LP relaxations was also identified in a concurrent work by Singer, Tulsiani, and Velusamy \cite{singer2025sketchingapproximationslpapproximations}. They conjecture that for every CSP, there must be a single-pass $o(n)$ space streaming algorithm that solves its LP relaxation, and any algorithm beating this factor must require at least $\Omega(n)$ space. For $\mdcut$, its basic LP achieves $1/2$-approximation and so their conjecture states that there is a single-pass algorithm that achieves $1/2$-approximation in $o(n)$ space.
In this work, we make progress towards this conjecture and prove it in the two-pass setting.

\begin{theorem}[Main result]\label{thm:main}
   For every $\epsilon>0$, there is a two-pass streaming algorithm that achieves $(1/2-\epsilon)$-approximation for $\mdcut$ on any $n$-vertex graph using space $O(n^{1-\Omega(1)})$.
\end{theorem}

\section*{Related work}
\paragraph*{Random order streaming.} Another related streaming model is the random order setting where the input is adversarial but the order in which it arrives in the stream is uniformly random. This additional randomness in the input has been successfully used to design better algorithms for many problems in the streaming setting \cite{DBLP:conf/icalp/MonemizadehMPS17,DBLP:conf/soda/PengS18,DBLP:conf/icalp/BravermanVWY18,DBLP:conf/soda/ChakrabartiG0V20,DBLP:conf/icalp/AssadiB21,DBLP:journals/mst/Bernstein24,velusamy2025optimalstreamingalgorithmdetecting}. In the context of CSPs, for $\mcut$, Kapralov, Khanna, and Sudan \cite{DBLP:conf/soda/KapralovKS15} show that any random order streaming algorithm achieving $(1/2+\epsilon)$-approximation requires at least $\Omega(\sqrt{n})$ space. 
This result was generalized by Saxena, Singer, Sudan, and Velusamy \cite{DBLP:conf/soda/SaxenaS0V23}, to show the approximation resistance of all CSPs that ``weakly support one-wise independence,'' in the random order setting.
On the algorithmic side, for $\mdcut$,  they give a $0.485$-approximation algorithm using $O(\log n)$ space. In a later paper \cite{DBLP:conf/soda/SaxenaS0V25}, they also give a $(1/2-\epsilon)$-approximation algorithm for $\mdcut$ for bounded-degree graphs, using $O(\log n)$ space.
\paragraph*{Decidability and exact computability.} The decidability problem for a CSP asks to output a bit deciding whether an input instance is fully satisfiable or not. Chou, Golovnev, Sudan, and Velusamy \cite{DBLP:journals/jacm/ChouGSV24} show that for every ``non-trivial'' CSP, any streaming algorithm that solves its decision version requires at least $\Omega(n)$ space. A result due to Kol, Paramonov, Saxena, and Yu \cite{DBLP:conf/innovations/KolPSY23} gives a dichotomy theorem in the multi-pass setting (for a constant number of passes) for exactly computing the value of a Boolean CSP. Their characterization is based on the Fourier degree of a Boolean CSP.
 
\section{Technical overview}
In this section, we present a technical overview of the proof of our main result.

There are three main ingredients to our proof: 1) a single-pass $(1/2-\epsilon)$-approximation algorithm of Saxena, Singer, Sudan, and Velusamy \cite{DBLP:conf/soda/SaxenaS0V25} for bounded-degree graphs that uses $O(n^{1-\Omega(1)})$ space, 2) Trevisan's reduction \cite{DBLP:conf/stoc/Trevisan01} from unbounded-degree to bounded-degree graphs that approximately preserves the $\mdcut$ value, and 3) Bhaskara, Dharuki, and Venkatasubramanian's \cite{DBLP:conf/icalp/BhaskaraDV18} two-pass algorithm that achieves $(1-\epsilon)$-approximation on graphs of average degree at least $n^\delta$ using $\tilde{O}(n^{1-\delta})$ space. The following theorems formally state their results.

\begin{theorem}[\cite{DBLP:conf/soda/SaxenaS0V25}]\label{thm:SSSV}
  For every $d \in \mathbb{N}$ and $\epsilon>0$, there is a single-pass streaming algorithm that achieves $(1/2-\epsilon)$-approximation for $\mdcut$ on any $n$-vertex graph using space $O(n^{1-\Omega(1)})$.
\end{theorem}

\begin{theorem}[\cite{DBLP:conf/stoc/Trevisan01}]\label{thm:Tre}
    For any directed graph $G(V,E)$ on $n$ vertices and $m$ edges, and $0<\epsilon\le 1$, there exists a (random) directed graph $\overline{G}(\overline{V},\overline{E})$ on $2m$ vertices and $O(m/\epsilon^2)$ edges such that with probability at least $5/6$, $\mdcut(G) - \epsilon' \le \mdcut(\overline{G}) \le \mdcut(G) + \epsilon$, and the maximum degree of a vertex in $\overline{G}$ is at most $O(1/\epsilon^2)$.
\end{theorem}

\begin{theorem}[\cite{DBLP:conf/icalp/BhaskaraDV18}]\label{thm:aditya}
For every $\epsilon, \delta\in (0,1]$, there is a two-pass streaming algorithm that achieves $(1-\epsilon)$-approximation for $\mdcut$ on any $n$-vertex graph of average degree at least $n^\delta$, using space $\tilde{O}(n^{1-\delta})$.
\end{theorem}

At a high level, our idea is to reduce the input graph which could potentially have an unbounded maximum degree to a bounded-degree graph using Trevisan's reduction \cite{DBLP:conf/stoc/Trevisan01}, and then execute the \cite{DBLP:conf/soda/SaxenaS0V25} algorithm for bounded-degree graphs. However, there are challenges to this naive approach. Firstly, Trevisan's reduction is not a small-space reduction and requires the knowledge of the degree of every vertex in the graph. This would be too expensive to store in the streaming setting. Secondly, the reduction produces a graph with $2m$ vertices, which could be $\Theta(n^2)$ in the worst case, hence the \cite{DBLP:conf/soda/SaxenaS0V25} algorithm could require $\Omega(n)$ space. We deal with the second challenge by running the \cite{DBLP:conf/soda/SaxenaS0V25} algorithm only when the input graph is sufficiently sparse, i.e., has at most $n^{1+\delta}$ edges, for some suitably chosen $\delta$ so that the space usage is still at most $O(n^{1-\Omega(1)})$. And when the input instance is dense, we run the \cite{DBLP:conf/icalp/BhaskaraDV18} algorithm. Thus, the space usage is always at most $O(n^{1-\Omega(1)})$. To deal with the first challenge, we implement the \cite{DBLP:conf/soda/SaxenaS0V25} algorithm in a non black-box manner. We do not produce the entire bounded-degree graph from the Trevisan reduction and instead only sample parts of this graph that are required for the \cite{DBLP:conf/soda/SaxenaS0V25} algorithm. Below, we give a high-level overview of both these algorithms and our implementation of them.

\paragraph*{Trevisan reduction \cite{DBLP:conf/stoc/Trevisan01}.} The reduction is quite simple. Say $G(V,E)$ is an unweighted directed graph on $n$ vertices and $m$ edges, with possible multi-edges but no self-loops. Given $G$, Trevisan first creates a weighted graph $G'(V',E')$: for every vertex $v\in V$, he adds $\deg(v)$ many ``copies'' of it, labeled $(v,0),\dots,(v,\deg(v)-1)$, to $V'$. For every edge $u\rightarrow v\in E$, and $i\in [\deg(u)], j\in [\deg(v)]$,\footnote{For a natural number $n\in \mathbb{N}$, we define $[n] = \{0,\dots,n-1\}$. Note that this is slightly different from the standard definition.} he adds an edge $(u,i)\rightarrow (v,j)$ of weight $\frac{1}{\deg(u)\deg(v)}$ to $E'$. It is not hard to see that the $\mdcut$ values of $G$ and $G'$ are equal. He then samples a random graph $\tilde{G}$ based on $G'$ in the following way. The vertex set of $\tilde{G}$ is $V'$ and he samples the edge set as follows: he sets a parameter $d=O(1/\epsilon^2)$ and samples $dm$ edges independently (with repetition) from $G'$, where in each sample, an edge in $E'$ is sampled with probability proportional to its weight. He finally deletes a minimal set of edges from $\tilde{G}$ so that the degree of every vertex is at most $O(d)$. He argues that with the high probability, at most an $\epsilon$ fraction of the edges are deleted, and that the $\mdcut$ value is preserved up to an additive $\epsilon$ factor.

\paragraph*{\cite{DBLP:conf/soda/SaxenaS0V25} algorithm.} Saxena, Singer, Sudan, and Velusamy make a key observation that on bounded-degree graphs, in order to simulate a distributed algorithm for $\mdcut$, it suffices to estimate certain neighborhood distributions of the input graph (see \cref{SSSV_nhbd_dist} for more details). They show that such distributions can be estimated from the induced subgraph of a randomly sampled subset of the vertex set. Finally, based on a previous distributed algorithm of Censor-Hiller, Levy, and Shachnai \cite{DBLP:conf/algosensors/Censor-HillelLS17}, they give a distributed algorithm that achieves $(1/2-\epsilon)$-approximation for $\mdcut$, and efficiently simulate it in sublinear space in the single-pass setting using the above techniques.

\paragraph*{Putting them together (a three-pass algorithm).} Since the \cite{DBLP:conf/soda/SaxenaS0V25} algorithm only requires the induced subgraph of a randomly sampled subset of the vertex set, we do not need to do the Trevisan reduction on the entire input graph. Instead, we directly sample a random subset of vertices from the bounded-degree graph, i.e., from $V' = \{(v,i): v\in V; i\in [\deg(v)]\}$, and perform a \emph{partial} Trevisan reduction on their \emph{parent} vertices. To sample from $V'$, we first map every edge $e\in E$, based on its position in the stream, to two unique vertices in $V'$, i.e., no two edges map to the same vertex. In particular, an edge $u\rightarrow v$ is mapped to $(u,i)$ and $(v,j)$ iff it is the $i$-th edge (resp. $j$-th edge) incident on $u$ (resp. $v$) in the stream. Thus, as the edges arrive in the stream, we can sample from $V'$.\footnote{There is a subtlety here. We do not have enough space to remember the ordering of the edges incident on every vertex in the input graph, which is crucial in the above mapping. Instead, we label the sampled copies of a vertex based on the order in which they are sampled by our algorithm, and argue that this relabeling does not affect the correctness of the algorithm (see \cref{sec:algorithms} for more details).} Once the vertices are sampled, we have to sample the induced subgraph on these vertices. The original Trevisan reduction performs a \emph{global} sampling of the edges, where every edge is sampled with probability proportional to its weight. In order to perform such a sampling, one would have compute the weights on all the edges, which is infeasible in the streaming setting. We therefore modify the reduction to a \emph{local} sampling procedure, where for every edge $u\rightarrow v\in E$, we create $d=O(1/\epsilon^2)$ ``copies'' of it in $E'$ and for every copy of the edge, we independently sample a uniform random copy of $u$ and a uniform random copy of $v$ in $V'$ as its endpoints.\footnote{A slightly different reduction, also amenable to streaming, was recently proposed by Fei, Minzer, and Wang in \cite{fei2025dichotomytheoremmultipassstreaming}.}We also delete \emph{all} the edges incident on high-degree vertices and not just the minimal set of edges, as it is not clear how to compute the minimal set without generating the entire graph. We argue in \cref{sec:modified_trevisan} that this procedure still preserves the $\mdcut$ value up to an additive $\epsilon$ factor with high probability. Now that we have this local procedure, we generate the induced subgraph on the sampled vertices as follows. After we sample the vertices in the first pass, we store the degrees of their parent vertices in the second pass. In the third pass, we again scan through the edges in the stream and if $u\rightarrow v$ is such that both $u$ and $v$ have at least one of their copies sampled, then we sample $d$ independent copies of the edge as outlined above (we can do this since we know the degrees of $u$ and $v$), and we store those edges that are in the induced subgraph. If neither $u$ nor $v$ have any copies in the sampled set of vertices, then we do nothing. If only one of them, say $u$, has sampled copies, then we sample only the endpoints corresponding to $u$'s copies for the $d$ copies of the edge. We do this to keep track of the degrees of all the sampled vertices, and in the end, delete all the edges incident on the high-degree vertices. The formal description of the procedure and the proof of correctness are described in \cref{sec:algorithms}.

\paragraph*{Modifications to obtain a two-pass algorithm.} In \cref{sec:twopass}, we modify the above algorithm and get a two-pass algorithm by getting rid of the second pass which computes and stores the degrees of the parent vertices \emph{exactly}. In particular, we show that our modified Trevisan reduction in \cref{sec:modified_trevisan} works even when we have only \emph{approximate} estimates for the degrees of high-degree vertices. In the first pass, along with sampling vertices like in the above algorithm, we estimate the degrees of high-degree vertices using standard sampling methods. In particular, for every vertex, we sample each edge incident to it independently with probability $p=n^{-O_{\epsilon}(1)}$ and rescale the total number of sampled edges by $1/p$ to estimate its degree. This procedure would estimate the degrees accurately (up to a multiplicative error) for vertices of sufficiently-large polynomial degree, say $n^c$. We can ensure that the space is still sublinear because we run the Trevisan procedure only on graphs with bounded average degree.
We then use these estimates to run edge sampling as in the third pass of the above algorithm. However, the copies of the edges incident on low-degree vertices will not be sampled correctly in this step. To fix this, for every vertex that is sampled in the first pass, we store the first $n^c$ edges incident on its \emph{parent} vertex in the second pass.\footnote{We need to be a bit careful when defining the degree thresholds and the sampling probabilities, but they can be set in a way to keep the space sublinear.} Thus, we can identify parent vertices of degree less than $n^c$, and also compute their degrees \emph{exactly}. In post-processing, we delete the edges incident to the copies of such vertices, and resample them according to their correct distributions (since we now have the exact degrees of the low-degree parent vertices).

\paragraph*{\cite{DBLP:conf/icalp/BhaskaraDV18} algorithm for $\mcut$.} Finally, we give a high-level description of the two-pass algorithm for dense instances of $\mcut$ in \cite{DBLP:conf/icalp/BhaskaraDV18}. While their paper explicitly states the result only for $\mcut$, their algorithm immediately extends to $\mdcut$ as well. Their main idea is to use certain non-uniform sampling methods to construct a ``core-set'' of \emph{sublinear} size for the input graph, which preserves the $\mcut$ value up to an additive error $\epsilon$. A core-set of linear size can be obtained for any graph by uniformly sampling $\tilde{O}(n/\epsilon^2)$ edges \cite{AG09,DBLP:conf/stoc/Trevisan09}. And for very dense graphs of average degree at least $\Omega(n)$, constant-size core-sets can be constructed \cite{DBLP:journals/jacm/GoldreichGR98,DBLP:journals/jcss/AlonVKK03,DBLP:conf/soda/MathieuS08}. In the polynomial density regime (i.e., average degree $n^\delta$ for $\delta<1$), previous works \cite{DBLP:journals/rsa/FeigeS02,DBLP:conf/soda/BarakHHS11} show that if both the maximum and the minimum degrees are $\Theta(n^\delta)$, then a core-set of size $\tilde{O}(n^{1-\delta})$ exists. Bhaskara, Dharuki, and Venkatasubramanian \cite{DBLP:conf/icalp/BhaskaraDV18} extend this result to general graphs in this regime, and give a two-pass streaming algorithm to construct a core-set of size $\tilde{O}(n^{1-\delta})$.

\section*{Bibliographic note}
In a concurrent and independent work~\cite{azarmehretal}, Azarmehr, Behnezhad, Ferrante, and Sanneian present a single-pass streaming algorithm that achieves a $(1/2-\epsilon)$-approximation using $O(n^{1-\Omega_\epsilon(1)})$ space. While their result is qualitatively stronger, our techniques appear to extend more naturally to arbitrary CSPs. While our approach relies only on a \emph{black-box} use of distributed algorithms for $\mdcut$, their algorithm is closely tailored to the distributed algorithm of Censor-Hillel, Levy, and Shachnai for $\mdcut$~\cite{DBLP:conf/algosensors/Censor-HillelLS17}. The latter crucially exploits the submodularity of $\mdcut$ and does not extend even to $\textsf{Max-$2$-AND}$.

More generally, Yoshida~\cite{Yos11} gives distributed algorithms for arbitrary CSPs based on packing LP formulations~\cite{KMW,li_et_al:LIPIcs.ISAAC.2024.45}. These algorithms imply single-pass sublinear-space streaming algorithms for bounded-degree CSP instances~\cite{singer2025sketchingapproximationslpapproximations}. Combining these results with the techniques developed in this work yields two-pass streaming algorithms for sparse CSP instances with polynomially bounded average degree.
This suggests the following natural open problem: extend the core-set construction of~\cite{DBLP:conf/icalp/BhaskaraDV18} to polynomially dense instances of general CSPs. Such a result would yield two-pass sublinear-space streaming algorithms for \emph{all} CSPs.

\section{Preliminaries}
In this paper, any reference to a directed graph $G$ means unweighted graphs with multi-edges but no self-loops. For a directed graph $G(V,E)$, a dicut is an ordered bipartition of the vertex set $V$. An edge $u\rightarrow v$ is said to belong to a dicut $L\sqcup R$ if $u\in L$ and $v\in R$. The value of the dicut, denoted by $\dcut(G, L\sqcup R)$ is the fraction of edges in $G$ that belong to the dicut $L\sqcup R$. We define $\mdcut(G)$ to be $\max_{L\sqcup R} \dcut(G, L\sqcup R)$.
For any function $\rho:V\rightarrow[0,1]$, we define the expected dicut value of $G$ under the rounding $\rho$, denoted by $\dcut(G,\rho)$, to be $\Exp[\dcut(G, L\sqcup R)]$, where every vertex in $v$ belongs to $L$ independently with probability $\rho(v)$.

For any function $f:A\rightarrow B$ that maps a domain $A$ to a co-domain $B$, and $C\subseteq A$, we denote the restriction of the function to the domain $C$ by $f|_C : C\rightarrow B$.

\subsection{Neighborhood distributions and the simulation of local algorithms}\label{SSSV_nhbd_dist}
A key contribution of Saxena, Singer, Sudan, and Velusamy in \cite{DBLP:conf/soda/SaxenaS0V25} is the identification of a connection between streaming and local algorithms, where they show that it is possible to estimate local neighborhood distributions of graphs in the streaming setting and use these distributions to simulate constant-round distributed algorithms. We recall some of the notations introduced in their paper and reprove their main theorem that shows that sampling vertices uniformly at random and estimating local neighborhood distributions from the induced subgraph suffices to estimate the value of a local algorithm that is executed on the original graph.

First, we formally define the LOCAL model for distributed algorithms for the $\mdcut$ problem and state the result of \cite{DBLP:conf/algosensors/Censor-HillelLS17, DBLP:conf/soda/SaxenaS0V25}.
\begin{definition}[The LOCAL model for distributed algorithms for $\mdcut$]
Let $G(V,E)$ be a directed graph on $n$ vertices and $m$ edges. $G$ describes the topology of a communication network, where each node can communicate with both its incoming and outgoing neighbors. Communication happens in synchronous rounds, where in every round, every node does some internal computations and sends an arbitrarily long message to each of its neighbors. At the end, each node $v$ computes a fractional assignment $\rho(v)\in [0,1]$ and the approximation ratio achieved by the algorithm is defined to be $\dcut(G,\rho)/ \mdcut(G)$.  In randomized algorithms, the nodes also have access to an arbitrarily long private random bit string that is sampled prior to any communication between the nodes.
\end{definition}

\begin{theorem}[\cite{DBLP:conf/algosensors/Censor-HillelLS17,DBLP:conf/soda/SaxenaS0V25}]
For every $\epsilon>0$, there exists a distributed algorithm for $\mdcut$ in the LOCAL model that achieves $(1/2-\epsilon)$-approximation with probability at least $9/10$, and requires at most $1/\epsilon$ rounds of communication. It is a randomized algorithm where every node uses only an $O(\log(1/\epsilon))$-bit private random string, and the random strings can be sampled from a pairwise independent hash family.
\end{theorem}

Saxena, Singer, Sudan, and Velusamy \cite{DBLP:conf/soda/SaxenaS0V25} observed that for each vertex $v$, $\rho(v)$ depends only on the structure of a ball of radius $1/\epsilon$ centered at that node, and the private random strings of all the nodes in this ball. They called this the ``neighborhood type'' of a vertex. They observed that in order to compute $\dcut(G,\rho) = (\sum_{u\rightarrow v\in E} \rho(u)(1-\rho(v)))/m$ with a small additive error and a constant success probability, it suffices to sample a random constant-sized subset $S\subset E$ and output $\sum_{u\rightarrow v\in S} \rho(u) (1-\rho(v))/|S|$. This motivated them to define the notion of a neighborhood type of an edge that captures the neighborhood types of both its endpoints. Thus, estimating the distribution of the edge-neighborhood types would suffice to estimate $(\sum_{u\rightarrow v\in E} \rho(u)(1-\rho(v)))/m$, which in turn gives a $(1/2-\epsilon)$-approximation to $\mdcut(G)$. For graphs whose maximum degree is bounded by a constant, the set of all possible edge-neighborhood types is constant sized, and hence, they could estimate their distribution efficiently in the streaming setting.
In what follows, we will set up some notations to define the neighborhood types and reprove a key theorem from \cite{DBLP:conf/soda/SaxenaS0V25}.

Let $G(V,E)$ be a directed graph and $u,v \in V$. We say that $v$ is of distance at most $\ell$ from $u$ if there exists a \emph{path} from $u$ to $v$ which is a sequence of vertices $u = w_0,w_1,\ldots,w_{\ell-1},w_\ell=v \in V$ such that for each $i \in [\ell]$, $w_{i-1}\rightarrow w_i \in E$ or $w_i\rightarrow w_{i-1} \in E$.
Let $\DbAllTypesDeg{c}{\ell}{D}$ denote the set of all isomorphism classes of doubly rooted labeled graphs where every vertex is of degree at most $D$, has a $c$-bit label, and is of distance at most $\ell$ from (at least) one of the roots. Consider a directed graph $G(V,E)$ whose maximum degree is at most $D$. Let $\mathcal{H}_n$ be the pairwise independent hash family that is used to sample the private random strings for every node in the distributed algorithm of \cite{DBLP:conf/algosensors/Censor-HillelLS17,DBLP:conf/soda/SaxenaS0V25}. Let $h:V\rightarrow \{0,1\}^c$ be any hash function in $\mathcal{H}_n$. For an edge $e\in E$, consider the following doubly rooted labeled graph $\mathcal{G}_e$: $\mathcal{G}_e$ is rooted at the endpoints of $e$ and consists of all the vertices in $G$ which are of distance at most $\ell$ from (at least) one of the roots, and their corresponding induced subgraph. In addition, each vertex $v$ in $G_e$ is labeled $h(v)$.
The radius-$\ell$ neighborhood type of $e$ with respect to $G$ and the hash function $h$, $\nbrtype{G}{h}{\ell}{e}$, is defined to be the representative of the isomorphism class in $\DbAllTypesDeg{c}{\ell}{D}$ that contains $\mathcal{G}_e$. The neighborhood edge-type distribution of $G$ with respect to the hash function $h$ is defined as follows.

\begin{definition}[Edge-type distribution]
    Let $h:V\rightarrow \{0,1\}^c$, $\ell,D \in \mathbb{N}$ and $G = (V,E)$ be a directed graph of maximum degree $D$. The \emph{radius-$\ell$ neighborhood type distribution} of $G$ with respect to $h$, denoted $\EdgeDist{G}{h}{\ell}$, is the distribution over $\DbAllTypesDeg{c}{\ell}{D}$ given by sampling a random $e \sim \Unif{E}$ and outputting $\nbrtype{G}{h}{\ell}{e}$.
\end{definition}

\begin{corollary}[\cite{DBLP:conf/algosensors/Censor-HillelLS17,DBLP:conf/soda/SaxenaS0V25}]\label{corollary_SSSV}
For every $\epsilon''>0$, there exist $\ell, c, n_0\in \mathbb{N}$ and a function $\Local : \DbAllTypesDeg{c}{\ell}{D} \to [0,1]$ such that for every $n\ge n_0$, there exists a hash family $\mathcal{H}_n$ such that the following holds. Let $G = (V,E)$ be a directed graph on $n$ vertices and of maximum degree $D$. Let $h\sim \Unif{\mathcal{H}_n}$. Then with probability at least $9/10$,
\[ \frac{\mdcut(G)}{2} - \epsilon'' \leq \Exp_{T \sim \EdgeDist{G}{h}{\ell}}[\Local(T)] \leq \mdcut(G). \]
\end{corollary}

The following corollary implies that it suffices to compute a distribution that is close in total variation distance to $\EdgeDist{G}{h}{\ell}$.
\begin{corollary}\label{corollary_TV_dist}
    If $\mathcal{D} \in \Delta(\DbAllTypesDeg{c}{\ell}{D})$ is such that $\tvdist{\mathcal{D}}{\EdgeDist{G}{h}{\ell}} \leq \epsilon''$, then \[ \frac{\mdcut(G)}{2} - 2\epsilon'' \leq \Exp_{T \sim \mathcal{D}}[\Local(T)]  \leq \mdcut(G) +  \epsilon''\, . \]
\end{corollary}

The following theorem (combined with \cref{corollary_TV_dist}) shows that for bounded-degree graphs, in order to achieve a $(1/2-\epsilon)$-approximation to $\mdcut$, it suffices to sample a sublinear number of random vertices and compute a ``re-scaled' edge-type neighborhood distribution of the induced subgraph.
\begin{theorem}[Simplified non-streaming version of Algorithm 2 in \cite{DBLP:conf/soda/SaxenaS0V25}]\label{main_thm_SSSV}
    For every $\epsilon''>0, \alpha\in [0,1), D,c,\ell\in \mathbb{N}$, there exist $\beta_0\in (0,1)$ and $n_0\in \mathbb{N}$ such that for all $n\ge n_0$ and $n^{-\beta_0} \le p \le 1$, the following hold: Let $G(V,E)$ be a directed graph on $n$ vertices and $m$ edges such that the maximum degree is at most $D$ and at most $\alpha$ fraction of the vertices are isolated. Let $h\in V\rightarrow\{0,1\}^c$. Let $V'\subset V$ be a random subset where every vertex in $V$ is sampled independently with probability $p$. Let $G'\subseteq G$ be the induced subgraph of $G$ on the vertices in $V'$. For every $T\in \DbAllTypesDeg{c}{\ell}{D}$, let $a_T$ denote the number of edges $e$ in $G'$ such that all the vertices in $G$ which are at a distance of at most $\ell$ from at least one of the endpoints of $e$, are contained in $V'$, and $\nbrtype{G'}{h|_{V'}}{\ell}{e} = T$. Consider the distribution $\mathcal{D}\in \Delta(\DbAllTypesDeg{c}{\ell}{D})$, defined to be
    \[\mathcal{D}(T) = \frac{a_T \cdot p^{-|T|}}{\sum_{T'\in \DbAllTypesDeg{c}{\ell}{D}}a_{T'} \cdot p^{-|T'|}} \, , \] where $T\in\DbAllTypesDeg{c}{\ell}{D}$ and $|T|$ denotes the number of vertices in $T$. Then, with probability at least $9/10$, $\tvdist{\mathcal{D}}{\EdgeDist{G}{h}{\ell}}\le \epsilon''$.
\end{theorem}

\begin{proof}
   Define $r_{c}^{\ell,D} \triangleq |\DbAllTypesDeg{c}{\ell}{D}|$. Fix $\beta_0 = \frac{1}{4D^\ell}$ and $n_0 = \left(\frac{20 (r_{c}^{\ell,D})^3 D^{2\ell+1}}{(\epsilon'')^2 (1-\alpha)}\right)^2$. We want to show that \[\sum_{T\in \DbAllTypesDeg{c}{\ell}{D}} |\mathcal{D}(T) -\EdgeDist{G}{h}{\ell}(T) |\le 2\epsilon''\, .\]
   Fix a type $T\in \DbAllTypesDeg{c}{\ell}{D}$. Let $S_T$ denote the set of all edges in $E$ such that $\nbrtype{G}{h}{\ell}{e} = T$. For $e\in S_T$, let $Y^T_e$ be the indicator random variable for the event that all vertices in $V$ which are at a distance of at most $\ell$ from at least one of the endpoints of $e$ belong to $V'$. Then, $a_T = \sum_{e\in S_T} Y_e^T$ because for $e\in S_T$, $\nbrtype{G}{h}{\ell}{e} = \nbrtype{G'}{h|_{V'}}{\ell}{e}$ iff $Y^T_e=1$.
   The expected value of $a_T$ is \[\Exp[a_T] = |S_T| \cdot p^{|T|} = \EdgeDist{G}{h}{\ell}(T)\cdot m \cdot p^{|T|}\, .\] We now compute the variance of $a_T$. Observe that
   \[\Var[a_T]\le \sum_{e,e'\in S_T: Y_e^T\text{ and }Y_{e'}^T\text{ are dependent}} \Exp[Y_e^T Y_{e'}^T] \, .\]
   Consider any edge $e\in S_T$. There are at most $2D^{2\ell+1}$ edges $e'$ such that $Y_{e'}^T$ and $Y_e^T$ are dependent. This is because $e$ and $e'$ are dependent iff they share an endpoint or there exists a vertex that is at a distance of at most $\ell$ from one of the endpoints of $e$ and one of the endpoints of $e'$. The bound follows because the maximum degree of any vertex in $G$ is at most $D$. 
   Since $Y_{e'}^T \le 1$, this implies that 
   \[\Var[a_T]\le 2D^{2\ell+1} \Exp[a_T] \, .\]
   Let $\delta = \epsilon''/r_{c}^{\ell,D}$. Applying Chebyshev's inequality (see \cref{thm:Chebyshev}), we have
   \[\Pr[|a_T-\Exp[a_T]|>\delta m p^{|T|}]\le \frac{2D^{2\ell+1} p^{- |T|}}{\delta^2 m} \le 1/(10r_{c}^{\ell,D})\, ,\] by the choice of $\beta_0$, $n_0$, and using the fact that $p\ge n^{-1/(2|T|)}$, $n\ge n_0$ and $m\ge n(1-\alpha)$.
   Taking the union bound over all possible types $T$, we have that with probability at least $9/10$, for every $T\in \DbAllTypesDeg{c}{\ell}{D}$, $|a_T-\Exp[a_T]|\le \delta m p^{|T|}$, and thus, \[\sum_{T\in \DbAllTypesDeg{c}{\ell}{D}}|a_T-\Exp[a_T]|\le \epsilon'' m p^{|T|} \, .\]

Then,
\begin{align*}
    &\sum_{T\in \DbAllTypesDeg{c}{\ell}{D}} |\mathcal{D}(T) -\EdgeDist{G}{h}{\ell}(T) | \\
    &= \sum_{T\in \DbAllTypesDeg{c}{\ell}{D}} \left|\frac{a_T p^{-|T|}}{m} -\EdgeDist{G}{h}{\ell}(T) +\mathcal{D}(T) - \frac{a_T p^{-|T|}}{m}\right| \\
    &\le           \sum_{T\in \DbAllTypesDeg{c}{\ell}{D}} \left|\frac{a_T p^{-|T|}}{m} -\EdgeDist{G}{h}{\ell}(T) \right|+ \left|\mathcal{D}(T) - \frac{a_T p^{-|T|}}{m}\right|                          \\
    &\le \epsilon'' + \sum_{T\in \DbAllTypesDeg{c}{\ell}{D}} \left|\frac{a_T  p^{- |T|}}{\sum_{T'\in \DbAllTypesDeg{c}{\ell}{D}}a_{T'}  p^{-|T'|}} - \frac{a_T p^{-|T|}}{m}\right|\\
    &= \epsilon'' + \sum_{T\in \DbAllTypesDeg{c}{\ell}{D}} \frac{a_T p^{-|T|}}{\sum_{T'\in \DbAllTypesDeg{c}{\ell}{D}}a_{T'}  p^{-|T'|}} \left|\frac{\sum_{T'\in \DbAllTypesDeg{c}{\ell}{D}}a_{T'}  p^{-|T'|}}{m}-1\right|\\
    &\le \epsilon'' + \sum_{T\in \DbAllTypesDeg{c}{\ell}{D}} \frac{a_T p^{-|T|}}{\sum_{T'\in \DbAllTypesDeg{c}{\ell}{D}}a_{T'} \cdot p^{-|T'|}} \sum_{T'\in \DbAllTypesDeg{c}{\ell}{D}}\left|\frac{a_{T'}  p^{-|T'|}}{m}-\EdgeDist{G}{h}{\ell}(T')\right| \le 2\epsilon''\, .
\end{align*}

\end{proof}

\subsection{Concentration inequalities}
In this section, we list some standard concentration bounds.

\begin{theorem}[Chebyshev's inequality]\label{thm:Chebyshev}
    Let $X$ be a random variable with variance $\sigma^2$. For any $a>0$, we have
    \[\Pr[|X-\Exp[X]|> a] \le \frac{\sigma^2}{a^2}\, .\]
\end{theorem}

\begin{theorem}[Chernoff bound]\label{prelim:chernoff}
Let $X_1,\dots,X_n$ be independent random variables taking values in $\{0,1\}$. Consider the sum of these random variables, $S_n = X_1 + \cdots + X_n$. Let $\mu = \Exp[S_n]$. Then for $R \ge 2 e \mu  $, \[\Pr[S_n \ge R]\le 2^{-R}\, .\]
\end{theorem}

\begin{theorem}[Hoeffding's inequality]\label{thm:hoeffding's}
Let $X_1,\dots,X_n$ be independent bounded random variables such that $a_i\le X_i\le b_i$. Consider the sum of these random variables, $S_n = X_1 + \cdots + X_n$. Then Hoeffding's theorem states that for all $t>0$,
\[\Pr[|S_n-\Exp[S_n]|\ge t]\le 2\exp\left(-\frac{2t^2}{\sum_{i=1}^n(b_i-a_i)^2}\right) \, .\]
\end{theorem}

\section{Reduction from unbounded-degree graphs to bounded-degree graphs}\label{sec:modified_trevisan}
In this section, we give a generalized version of the Trevisan reduction \cite{DBLP:conf/stoc/Trevisan01} that is tailored to be implementable in the two-pass streaming setting in \cref{sec:twopass}.

\begin{algorithm}[H]
\caption{Algorithm for reducing unbounded-degree graphs to bounded-degree graphs}\label{Alg:trevisan}
\begin{algorithmic}[1]
\Statex \textbf{Parameters:} $n, m \in \mathbb{N}$, $\epsilon' > 0$, $\zeta>0$
\Statex \textbf{Input:} 
\Statex A directed graph $G(V,E)$ on $n$ vertices and $m$ edges
\Statex An $n$-length array $\appdeg$ (indexed by $V$) such that when $\deg(v)< n^\zeta$, we have $\appdeg(v)=\deg(v)$ and when $\deg(v)\ge n^\zeta$, we have $(1-\epsilon'/100)\deg(v) \le \appdeg[v]\le (1+\epsilon'/100)\deg(v)$
\Statex

\State Initialize directed graphs $\tilde{G}(\tilde{V}, \tilde{E}), \overline{G}(\overline{V},\overline{E})$, where $\tilde{V}=\overline{V} =\{(v,i):v\in V, i\in [\deg(v)]\} $ and $\tilde{E}=\emptyset,\overline{E}=\emptyset$
\State Set a parameter $d=80/(\epsilon')^2$
\For{every edge $u\rightarrow v\in E$}
\State Initialize $\mathsf{i}=1$
\For {$\mathsf{i}\le d$}
\State Sample $i_1\sim \Unif{[\appdeg(u)]}$ \label{alg_samplei1_trev}
\State Sample $i_2\sim \Unif{[\appdeg(v)]}$ \label{alg_samplei2_trev}
\If{$i_1<\deg(u)$ and $i_2<\deg(v)$} \label{check_deg}
\State Add $(u,i_1)\rightarrow(v,i_2)$ to $\tilde{E}$
\EndIf
\State Increment $\mathsf{i}$ by $1$
\EndFor
\EndFor
\State Copy all the edges in $\tilde{E}$ to $\overline{E}$
\For{every $v\in \overline{V}$ such that $\deg(v)> 11 d$}
\State Delete all the edges incident on $v$ in $\overline{E}$
\EndFor
\State Output $\overline{G}(\overline{V},\overline{E})$
\end{algorithmic}
\end{algorithm}

\begin{theorem}\label{thm:trevisan_reduction}
    Let $n\ge 2, m\ge 10\in \mathbb{N}$. For any directed graph $G(V,E)$ on $n$ vertices and $m$ edges, and $0<\epsilon'\le 1$, $\zeta>0$, with probability at least $5/6$, \cref{Alg:trevisan} outputs a directed graph $\overline{G}(\overline{V},\overline{E})$ on $2m$ vertices and $\Theta(m/(\epsilon')^2)$ edges such that , $\mdcut(G) - \epsilon' \le \mdcut(\overline{G}) \le \mdcut(G) + \epsilon'$, and the maximum degree of a vertex in $\overline{G}$ is at most $O(1/(\epsilon')^2)$.
\end{theorem}

\begin{proof}
Recall the construction of $\tilde{G}$ in \cref{Alg:trevisan}. We will first argue that the number of edges in $\tilde{E}$ is at least $md(1-\epsilon'/8)$, with high probability. Observe that for any edge $u\rightarrow v\in E$, when we sample its $\mathsf{i}$-th copy in $\tilde{G}$, the probability that either $i_1$ or $i_2$ fail to satisfy the condition in \cref{check_deg} is at most $\epsilon'/50$. Therefore, applying Chernoff bound (\cref{prelim:chernoff}), we conclude that with probability at least $1-\exp(-10m)$, $|\tilde{E}|\ge md(1-\epsilon'/8)$.

\paragraph*{$\mdcut$ value of $\tilde{G}$.}  Fix any ordered bipartition $L \sqcup R$ of $\tilde{V}$. For an edge $e\in E$, let $X_{e,\mathsf{i}}$ be the indicator random variable for the event that in the $i$-th iteration corresponding to edge $e$, the condition in \cref{check_deg} is satisfied and the corresponding edge in $\tilde{E}$ belongs to the dicut $L\sqcup R$.  Then, the $\dcut$ value of $\tilde{G}$ under $L\sqcup R$ is given by \[\dcut(\tilde{G}, L\sqcup R)  = \frac{\sum_{e\in E, 1\le \mathsf{i}\le d} X_{e,\mathsf{i}}}{|\tilde{E}| }\, .\] 
Let us define a function $\rho_{L\sqcup R}:V\rightarrow [0,1]$ such that for $v\in V$, $\rho_{L\sqcup R}(v) = (|L\cap \{(v,i):i\in [\deg(v)]\}|)/\deg(v)$, i.e., $\rho_{L\sqcup R}(v)$ is the fraction of copies of $v$ in $L$.
Then,
\[ \dcut(G,\rho_{L\sqcup R}) - \epsilon'/33 \le  \frac{\sum_{e\in E, 1\le \mathsf{i}\le d} \Exp[X_{e,\mathsf{i}}]}{dm} \le   \dcut(G,\rho_{L\sqcup R})+ \epsilon'/33\, .\]
This is because for every $u\rightarrow v \in E, 1\le \mathsf{i}\le d$, \begin{equation} \label{eqn:rho}
\rho_{L\sqcup R}(u)(1-\rho_{L\sqcup R}(v)) - \epsilon'/33 \le X_{u\rightarrow v,\mathsf{i}} \le \rho_{L\sqcup R}(u)(1-\rho_{L\sqcup R}(v)) + \epsilon'/33\, ,\end{equation} and 
\[\dcut(G,\rho) = \frac{\sum_{u\rightarrow v} \rho(u)(1-\rho(v))}{m}\, .\]
To see why \cref{eqn:rho} is true, observe that for $i_1\sim \Unif{[\appdeg[u]]}$, we have \[\rho_{L\sqcup R}(u)-\epsilon'/99\le \Pr[i_1<\deg(u)\text{ and }(u,i_1)\in L]\le \rho_{L\sqcup R}(u) + \epsilon'/99 \, ,\]since for every $u\in V$, $\appdeg[u]$ satisfies \[(1-\epsilon'/100)\deg(u) \le \appdeg[u]\le (1+\epsilon'/100)\deg(u)\, ,\]
Since all the edges were sampled independently and $d = 80/(\epsilon')^2$, we can apply Hoeffding's bound (\cref{thm:hoeffding's}) to conclude that with probability at least $1-\exp(-4m)$, \[\left|\frac{\sum_{e\in E, 1\le \mathsf{i}\le d} X_{e,\mathsf{i}}}{dm} -  \dcut(G,\rho_{L\sqcup R})\right|\le \epsilon'/4\, .\]
Conditioned on $|\tilde{E}|\ge md(1-\epsilon'/8)$, we have
\[\left|\dcut(\tilde{G}, L\sqcup R) -  \dcut(G,\rho_{L\sqcup R})\right|\le \epsilon'/2\, .\]
Since $\mdcut(G) = \max_{\rho:V\rightarrow\{0,1\}} \dcut(G,\rho) = \sup_{\rho:V\rightarrow[0,1]} \dcut(G,\rho)$, taking the union bound over the $2^{2m}$ possible partitions of $\tilde{V}$,\footnote{In particular, consider the $\mdcut$ partition $\rho^*:V\rightarrow \{0,1\}$ of $G$ and the partition $L\sqcup R$ of $\tilde{V}$ where $(v,i)\in L$ iff $\rho^*(v)=1$. By definition, $\rho_{L\sqcup R} = \rho^*$.} we conclude that with probability at least $1-\exp(-2m)$, \[|\mdcut(\tilde{G}) -  \mdcut(G)|\le \epsilon'/2\, .\]

\paragraph*{The total number of edges incident on high-degree vertices in $\tilde{G}$.} For a vertex $\tilde{v}\in \tilde{V}$, let $Y_{\tilde{v}} = \mathbbm{1}[\deg(\tilde{v}) > 11 d] \deg(\tilde{v})$. Then $Y = \sum_{\tilde{v}\in \tilde{V}} Y_{\tilde{v}}$ denotes the total sum of the degrees of vertices whose degree is larger than $11 d$. It suffices to show that for any $\tilde{v}$, $\Exp[Y_{\tilde{v}}]<1$. This is because applying Markov's inequality, we would get that $Y < 8m/7 < \epsilon' |\tilde{E}|/8$ with probability at least $7/8$, conditioned on $|\tilde{E}|\ge md(1-\epsilon'/8)$. Since $\mdcut(\tilde{G})\ge |\tilde{E}|/4$, the graph $\overline{G}$ obtained by deleting all the edges incident on vertices of degree larger than $11 d$ in $\tilde{G}$ would satisfy $|\mdcut(\overline{G}) -  \mdcut(\tilde{G})|\le \epsilon'/2$. We also get that $|\overline{E}|\ge md(1-\epsilon')/4$.

It remains to show that $\Exp[Y_{\tilde{v}}]<1$. We will argue that for $d'>11d$, $\Pr[\deg(\tilde{v}) \ge d']\le 2^{-d'}$ (thus, $\Exp[Y_{\tilde{v}}] \le \sum_{d'>11d} \frac{d'}{2^{-d'}} < 1$). 
Consider any vertex $\tilde{v}=(v,i)\in V$.
Consider any edge $e\in E$ that is incident on $v$. For $1\le j\le d$, let $Z_{e,j}$ be the indicator random variable for the event that in the $j$-th iteration corresponding to $e$, the condition in \cref{check_deg} is satisfied and the edge that is sampled is incident on $\tilde{v}$. We have \[\deg(\tilde{v}) = \sum_{e\in E:e\text{ is incident on }v}\sum_{1\le j\le d} Z_{e,j}\, .\]
Observe that $\Exp[Z_{e,j}] \le \frac{1}{\deg(v)(1-\epsilon'/100)}$
Thus, by linearity of expectation, $\Exp[\deg(\tilde{v})] \le \frac{d\deg(v)}{\deg(v)(1-\epsilon'/100)}\le 2d$ and since the edges were sampled independently, we can apply Chernoff bound (\cref{prelim:chernoff}) to conclude that $\Pr[\deg(\tilde{v}) \ge d']\le 2^{-d'}$.
\end{proof}

\section{Three-pass streaming algorithm for graphs of bounded average degree}\label{sec:algorithms}

In this section, for every constant $\epsilon>0$, we give a three-pass streaming algorithm that achieves $(1/2-\epsilon)$-approximation on any $n$-vertex graph with average degree \emph{at most} $n^{\delta}$, using space $O(n^{1- \delta})$, for some $\delta\in (0,1)$.

\begin{breakablealgorithm}
\caption{Three-pass streaming algorithm for $\mdcut$ on graphs of bounded average degree}\label{Alg:bounded_avg_degree_threepass}
\begin{algorithmic}[1]
\Statex \textbf{Parameters:} $n \in \mathbb{N}$, $\epsilon \in (0,1/2)$
\Statex \textbf{Input:} A stream $\vecsigma$ of edges of directed graph $G(V,E)$ on $n$ vertices and at most $n^{1+\delta}$ edges, where $\delta = \epsilon^{25/\epsilon}$

\Statex
\Statex \textbf{First pass:}
\State Maintain a counter for the number of edges $m$
\State Set a parameter $\beta \triangleq \epsilon^{20/\epsilon}$ \label{setting beta}
\State Initialize $V', E'=\emptyset$ 
\State Initialize an $n$-length array $\cnt$ (indexed by $V$) to be $0^n$ \label{count array}
\For{every edge $u\rightarrow v$ in the stream} 
\State Sample $z_1,z_2 \sim \Bern(n^{-\beta})$\footnote{We denote by $\Bern(p)$ the Bernoulli distribution where we draw $1$ with probability $p$ and $0$ with probability $1-p$.} independently \label{sampling_vertices_1}
\If {$z_1 =1$}
\State Add $(u,\cnt[u]$) to $V'$ and increment $\cnt[u]$ by $1$ \label{sampling_vertices_2}
\EndIf
\If {$z_2 =1$}
\State Add $(v,\cnt[v]$) to $V'$ and increment $\cnt[v]$ by $1$ \label{sampling_vertices_3}
\EndIf
\If{$|V'|>n^{1-3\beta/4}$} \label{terminate1}
\State Terminate the execution
\EndIf
\EndFor
\Statex

\Statex \textbf{Second pass:}
\If{$m\le n^{1-\epsilon^{1/\epsilon}}$}
\State Store all the edges and compute the $\mdcut$ value.
\Else
\For{every $v\in V$ such that $\cnt[v] > 0$}
\State Compute $\deg(v)$ in the stream
\State $\forall i \in [\cnt[v]]$, initialize $d(v,i)=0$
\EndFor
\EndIf
\Statex

\Statex \textbf{Third pass:}
\State Set a parameter $d=320/\epsilon^2$
\For{every edge $u\rightarrow v$ in the stream}
\State Initialize $\mathsf{i}=1$
\For {$\mathsf{i}\le d$}
\If{$\cnt[u]>0$} \label{edge__sampling_1}
\State Sample $i_1\sim \Unif{[\deg(u)]}$
\If {$i_1< \cnt[u]$}
\State Increment $d(u,i_1)$ by $1$ 
\EndIf
\EndIf
\If{$\cnt[v]>0$} \label{edge__sampling_2}
\State Sample $i_2\sim \Unif{[\deg(v)]}$
\If {$i_2< \cnt[v]$}
\State Increment $d(v,i_2)$ by $1$
\EndIf
\EndIf
\If{$\cnt[u]>0, \cnt[v]>0$ and $i_1< \cnt[u], i_2< \cnt[v]$}
\State Add $(u,i_1)\rightarrow(v,i_2)$ to $E'$ \label{adding_edge}
\If{$|E'|>n^{1-\epsilon^{25/\epsilon}}$} \label{terminate2}
\State Terminate the execution
\EndIf
\EndIf
\State Increment $\mathsf{i}$ by $1$
\EndFor
\EndFor
\Statex

\end{algorithmic}
\end{breakablealgorithm}

\paragraph*{Post processing steps of \cref{Alg:bounded_avg_degree_threepass}:}
\begin{enumerate}
    \item For every $(u,i)\in V'$ such that $d(u,i) > 11d=3520/\epsilon^2$, delete every edge incident on $(u,i)$ in $E'$ and decrement $d(\cdot)$ each of its endpoints by $1$. Finally, set $d(u,i) = 0$. Let $\mathcal{G}(V',E')$ be the directed graph on vertex set $V'$ and edge set $E'$. 
    \item  Let $\ell, c\in \mathbb{N}$ be set according to \cref{corollary_SSSV} for parameter $\epsilon''=\epsilon/8$. Let $D=11d$. Let $\mathcal{H}_{2m}$ be the hash family in \cref{corollary_SSSV} corresponding to directed graphs on $2m$ vertices. Sample $h\sim \Unif{\mathcal{H}_{2m}}$.
    \item For $T\in \DbAllTypesDeg{c}{\ell}{D}$, let $a_T$ denote the number of edges $e\in E'$ such that every vertex $(v,i)\in V'$ which is at a distance of at most $\ell-1$ from at least one of the endpoints of $e$ satisfies the property that $\deg_{\mathcal{G}}(v,i) = d(v,i)$, and $\nbrtype{\mathcal{G}}{h|_{V'}}{\ell}{e} = T$.
    \item Consider the distribution $\mathcal{D}\in \Delta(\DbAllTypesDeg{c}{\ell}{D})$, defined to be
    \[\mathcal{D}(T) = \frac{a_T \cdot n^{\beta |T|}}{\sum_{T'\in \DbAllTypesDeg{c}{\ell}{D}}a_{T'} \cdot n^{\beta|T'|}} \, , \] where $T\in\DbAllTypesDeg{c}{\ell}{D}$ and $|T|$ denotes the number of vertices in $T$.
    \item Let $\Local : \DbAllTypesDeg{c}{\ell}{D} \to [0,1]$ be the function in \cref{corollary_SSSV}. Output $\Exp_{T \sim \mathcal{D}}[\Local(T)]$.
\end{enumerate}

\begin{theorem}\label{thm: correctness of algorithm 1}
    For every $\epsilon\in (0,1/2)$, there exists $\mathsf{n}_0\in \mathbb{N}$ such that for all $n\ge \mathsf{n}_0$ and every directed graph $G$ on $n$ vertices and $m \le n^{1+\epsilon^{25/\epsilon}}$ edges, \cref{Alg:bounded_avg_degree_threepass} outputs a value $v$ satisfying
    \[
    \frac{\mdcut(G)}{2} - \epsilon \le v \le \mdcut(G) + \epsilon\, ,
    \] with probability at least $2/3$, using space $n^{1-\epsilon^{O(1/\epsilon)}}$.
\end{theorem}

\begin{proof}
We first analyze the space usage of \cref{Alg:bounded_avg_degree_threepass} and then argue its correctness.

\paragraph*{Space usage of \cref{Alg:bounded_avg_degree_threepass}.} It follows from \cref{terminate1,terminate2} that $|V'|,|E'| \le n^{1-\epsilon^{O(1/\epsilon)}}$. In \cref{count array}, we can maintain $\cnt$ as a linked list instead of an array, thus, reducing the memory required to store the non-zero entries in $\cnt$, which is at most $|V'|$. We initialized $\cnt$ as an $n$-length array to keep the presentation simple.

\paragraph*{Correctness of \cref{Alg:bounded_avg_degree_threepass}.} Consider the bounded-degree graph $\overline{G}$ that we constructed in \cref{thm:trevisan_reduction} such that \[\mdcut(G)-\epsilon/2 \le \mdcut(\overline{G}) \le \mdcut(G) + \epsilon/2\, .\] We will show that with probability at least 2/3, \[\tvdist{\mathcal{D}}{\EdgeDist{\overline{G}}{h}{\ell}}\le \epsilon/8\, .\]Therefore, combining \cref{thm:trevisan_reduction,corollary_SSSV,corollary_TV_dist}, the correctness of \cref{Alg:bounded_avg_degree_threepass} immediately follows.
Let us first recall the construction of $\overline{G}$ in \cref{Alg:trevisan} for the parameter $\epsilon' = \epsilon/2$, and when $\estdeg[v]=\deg(v)$ for every $v\in V$. For every vertex $v\in V$, we add $\deg(v)$ many copies of it to $\overline{V}$, i.e., $\overline{V}=\{(v,i):v\in V, i\in [\deg(v)]\}$. We fix a parameter $d=10/(\epsilon')^2 = 40/\epsilon^2$. For every edge $u\to v\in E$, we sample a copy $(u,i)$ of $u$ and a copy $(v,j)$ of $v$ uniformly at random from $\overline{V}$, and add the edge $(u,i)\to (v,j)$ to $\overline{E}$. we repeat this process $d$ times independently for every edge, and finally, for every vertex of degree larger than $11 d$ in $\overline{V}$, we delete all the edges incident to it. We prove that with probability at least $5/6$, \[\mdcut(G) - \epsilon/2 \le \mdcut(\overline{G}) \le \mdcut(G) + \epsilon/2\, .\]

For now, let us ignore the termination conditions in \cref{terminate1,terminate2}. We will later argue that these events occur with very low probability and hence, do not affect the success probability of the algorithm by much. We claim the following.

\begin{claim}\label{claim_alg}
At the end of the third pass (before post processing), the following hold:
    \begin{enumerate}
    \item $V'$ has the same distribution as sampling every vertex from $\overline{V}$ independently with probability $n^{-\beta}$, and applying a relabeling on the vertices, and \label{item 1_claim_alg}
    \item for $u\rightarrow v\in E$, if both $u$ and $v$ have copies in $V'$, then the copies of $u\rightarrow v$ in $\tilde{G}$ are sampled identically in \cref{Alg:trevisan,Alg:bounded_avg_degree_threepass}. If only one of $u$ or $v$ has a copy in $V'$, then the endpoints corresponding to the copies of that vertex in the copies of $u\rightarrow v$ in $\tilde{G}$ are sampled identically in \cref{Alg:trevisan,Alg:bounded_avg_degree_threepass}.\label{item 2_claim_alg}
\end{enumerate} 
\end{claim}

\begin{proof}
Consider the map $f:\overline{V}\rightarrow E$ defined as follows: for $(v,i)\in \overline{V}$, $f(v,i)=e$, where $e$ is the $(i+1)$-th edge incident on $v$ in the stream $\vecsigma$. Observe that for every edge $u\rightarrow v$, there are exactly two vertices in $\overline{V}$ that map to it, namely $(u,i)$ and $(v,j)$ where $i,j$ are such that $e$ is the $(i+1)$-th edge (resp. $(j+1)$-th edge) incident on $u$ (resp. $v$) in the stream $\vecsigma$. Therefore, sampling every vertex in $\overline{V}$ independently with probability $n^{-\beta}$ is the same as sampling each vertex in $f^{-1}(e)$, for every edge $e$, independently with probability $n^{-\beta}$. This is exactly what we do in \cref{sampling_vertices_1,sampling_vertices_2,sampling_vertices_3} in \cref{Alg:bounded_avg_degree_threepass}. However, since we cannot maintain information about the streaming order of all the edges incident on every vertex in $V$, we relabel the copies of $v$ in $\overline{V}$ based on the order in which they are sampled by our algorithm. In particular, for a vertex $v\in V$, the first copy that is sampled is labeled $(v,0)$, the second copy that is sampled is labeled $(v,1)$, and so on. This proves \cref{item 1_claim_alg} of \cref{claim_alg}. To see \cref{item 2_claim_alg} of \cref{claim_alg}, recall how the edges were sampled for $\tilde{G}$ in \cref{Alg:trevisan}. For our purposes, it suffices to consider the sampling of the copies of only those edges in $E$ that are incident on a vertex $v\in V$ which has a copy in $V'$, i.e., $V'\cap \{(v,i):i\in [\deg(v)]\}\ne \emptyset$. \cref{edge__sampling_1,edge__sampling_2} check for this condition and if they are satisfied, the sampling is identical to \cref{alg_samplei1_trev,alg_samplei2_trev} of \cref{Alg:trevisan}. 
For an edge $u\rightarrow v$, its copy is added to $E'$ in \cref{adding_edge} iff the sampled copies of both its endpoints are in $V'$. Observe that for $(v,i)\in V'$, all the edges incident on $v$ in $E$ satisfy either the condition in \cref{edge__sampling_1} or \cref{edge__sampling_2}.
\end{proof}
Given the above claim, we consider $\tilde{G}$ where for $u\rightarrow v\in E$, if both $u$ and $v$ have copies in $V'$, then the copies of $u\rightarrow v$ are sampled according to \cref{Alg:bounded_avg_degree_threepass} (prior to the post-processing steps). If only one of $u$ or $v$ has a copy in $V'$, then the endpoints corresponding to the copies of that vertex in the copies of $u\rightarrow v$ are sampled according to \cref{Alg:bounded_avg_degree_threepass} and the other endpoints are sampled independently, according to \cref{Alg:trevisan}. The remaining edges are also sampled according to \cref{Alg:trevisan}.
Finally, like in \cref{Alg:trevisan}, we define $\overline{G}$ as $\tilde{G}$ with edges incident on vertices of degree larger than $11d$ deleted.\footnote{A key observation here is that sampling $\overline{G}$ first and then sampling $V'$ is equivalent to sampling $V'$ and then sampling $\overline{G}$.}
Observe that for $\overline{G}$ as defined above, $\mathcal{G}(V',E')$ is the induced subgraph of $\overline{G}$ on $V'$ and furthermore, at the end of the post-processing steps of \cref{Alg:bounded_avg_degree_threepass}, for $(v,i)\in V'$, $d(v,i)$ is degree of $(v,i)$ in $\overline{G}$.  This is true because for any vertex $(v,i)\in V'$, its degree in $\tilde{G}$ depends only on the randomness in \cref{Alg:bounded_avg_degree_threepass}.

We are now ready to show that with probability at least $3/4$, \[\mdcut(G)-\epsilon/2 \le \mdcut(\overline{G}) \le \mdcut(G) + \epsilon/2\, ,\]and \[\tvdist{\mathcal{D}}{\EdgeDist{\overline{G}}{h}{\ell}}\le \epsilon/8\, .\]
We apply \cref{thm:trevisan_reduction,main_thm_SSSV}. In particular, for \cref{thm:trevisan_reduction}, we set the parameter $\epsilon'=\epsilon/2$. 
We get that with probability at least $5/6$,
\begin{itemize}
    \item $\mdcut(G)-\epsilon/2 \le \mdcut(\overline{G}) \le \mdcut(G) + \epsilon/2$,
    \item the maximum degree of any vertex in $\overline{G}$ is at most $110/(\epsilon')^2 = 3520/\epsilon^2$, and
    \item $\overline{G}$ has at least $240 m/\epsilon^2$ edges.
\end{itemize}
Since $\overline{G}$ has exactly $2m$ vertices, it follows that at most $0.94$ fraction of them are isolated. Now conditioned on this $\overline{G}$, we set the parameters of \cref{main_thm_SSSV}. We set $\epsilon''=\epsilon/8, D=3520/\epsilon^2$, and $c,\ell,h$ according to \cref{Alg:bounded_avg_degree_threepass}, and $\alpha=0.94$. Now set $\beta_0$ and $n_0$ according to the proof of \cref{main_thm_SSSV}. It follows from \cref{claim_alg} that $V'$ is obtained by sampling every vertex in $\overline{V}$ independently with probability $n^{-\beta}$, and $\mathcal{G}(V',E')$ is the induced subgraph of $\overline{G}$ on $V'$. Say $2m\ge n_0$. To apply \cref{main_thm_SSSV}, observe that sampling $\overline{G}$ first and then sampling $V'$ is the same as sampling $V'$ first and then sampling $\overline{G}$, and relabeling of the copies of the vertices does not the change the distribution.
Thus, it remains to show the following.
\begin{enumerate}
    \item The sampling probability $n^{-\beta}$ is at least $(2m)^{-\beta_0}$, and
    \item $a_T$ defined in the third post-processing step of \cref{Alg:bounded_avg_degree_threepass} is the same as the $a_T$ defined in \cref{main_thm_SSSV}.
\end{enumerate}

If $m\le n^{1-\epsilon^{1/\epsilon}}$, then \cref{Alg:bounded_avg_degree_threepass} would simply store all the edges and compute the $\mdcut$ value. So without loss of generality, we can assume that $m>n^{1-\epsilon^{1/\epsilon}}$, and in this case,  $n^{-\beta}\ge (2m)^{-\beta_0}$. In \cref{claim_alg}, we also show that for every $(v,i)\in V'$, $d(v,i)$ is the degree of $(v,i)$ in $\overline{G}$. Hence, for any edge $e\in E'$, its radius-$\ell$ neighborhood in $\overline{G}$ is contained entirely within $\mathcal{G}$ iff for every vertex $(v,i)\in V'$ which is at a distance of at most $\ell-1$ from at least one of the endpoints of $e$ satisfies $\deg_{\mathcal{G}}(v,i) = d(v,i)$. It follows that the $a_T$ defined in the post-processing steps of \cref{Alg:bounded_avg_degree_threepass} is the same as the one in \cref{main_thm_SSSV}. Thus, conditioned on $\overline{G}$ satisfying the properties that we stated above and $2m\ge n_0$, with probability at least $9/10$,
\[\tvdist{\mathcal{D}}{\EdgeDist{\overline{G}}{h}{\ell}}\le \epsilon/8\, .\]
Set $\mathsf{n}_0 = \max\left\{n_0^{(1-\epsilon^{1/\epsilon})^{-1}},(\frac{1}{\epsilon^2})^{{\frac{1}{\epsilon}}^{\frac{25}{\epsilon}}}\right\}$. Thus, whenever $n\ge \mathsf{n_0}$, we have $m\ge n_0$.

We finally argue that the termination conditions in \cref{terminate1,terminate2} happen with very low probability.
Since $m\le n^{1+\epsilon^{25/\epsilon}}$ and $|\overline{V}|=2m$, we have $\Exp[V']\le n^{1-7\beta/8}$. By applying Chernoff bound, we conclude that with very high probability, $|V'| \le n^{1-3\beta/4} $.
For an edge $e\in \overline{E}$, let $Z_e$ be the indicator random variable for the event that $e\in E'$. Let $Z=\sum_{e\in \overline{E}} Z_e$.
We have $\Exp[Z_e] = n^{-2\beta} = n^{-2\epsilon^{20/\epsilon}}$ and $\Var[Z]\le (3520/\epsilon^2) \mathbb{E}[Z]$, since the maximum degree of any vertex in $\overline{G}$ is at most $3520/\epsilon^2$. Since $|\overline{E}|\le 320m/\epsilon^2$, $\Exp[Z] \le n^{1-\epsilon^{22/\epsilon}}$. By applying Chebyshev's inequality, we conclude that with high probability, $|E'|\le n^{1-\epsilon^{25/\epsilon}}$.
This concludes the proof of correctness of \cref{Alg:bounded_avg_degree_threepass}.
\end{proof}

\section{Two-pass streaming algorithm for $\mdcut$}\label{sec:twopass}
In this section, we modify\footnote{The specific modifications are highlighted in {\color{red} red}.} \cref{Alg:bounded_avg_degree_threepass} to give a two-pass streaming algorithm that achieves $(1/2-\epsilon)$-approximation on any $n$-vertex graph with average degree \emph{at most} $n^{\delta}$, using space $O(n^{1- \delta})$, for some $\delta\in (0,1)$. Combining our algorithm with the two-pass algorithm of Bhaskara, Daruki, and Venkatasubramanian \cite{DBLP:conf/icalp/BhaskaraDV18} (see \cref{thm:aditya}) that achieves for every $\epsilon, \delta\in (0,1)$, $(1-\epsilon)$-approximation for graphs of average degree \emph{at least} $n^{\delta}$ using space $\tilde{O}(n^{1-\delta})$, we obtain our algorithm that works on arbitrary graphs (see \cref{Alg:meta}).
\begin{breakablealgorithm}
\caption{Two-pass streaming algorithm for $\mdcut$ on graphs of bounded average degree}\label{Alg:bounded_avg_degree_twopass}
\begin{algorithmic}[1]
\Statex \textbf{Parameters:} $n \in \mathbb{N}$, $\epsilon \in (0,1/2)$
\Statex \textbf{Input:} A stream $\vecsigma$ of edges of directed graph $G(V,E)$ on $n$ vertices and at most $n^{1+\delta}$ edges, where $\delta = \epsilon^{25/\epsilon}$

\Statex
\Statex \textbf{First pass:}
\State Maintain a counter for the number of edges $m$
\State Set a parameter $\beta \triangleq \epsilon^{20/\epsilon}$
\State Initialize $V', E'=\emptyset$ 
\State Initialize an $n$-length array $\cnt$ (indexed by $V$) to be $0^n$ \label{count array}
\State {\color{red} Initialize an $n$-length array $\estdeg$ (indexed by $V$) to be $0^n$}
\For{every edge $u\rightarrow v$ in the stream} 
\State Sample $z_1,z_2 \sim \Bern(n^{-\beta})$ independently \label{sampling_vertices_1}
\If {$z_1 =1$}
\State Add $(u,\cnt[u]$) to $V'$ and increment $\cnt[u]$ by $1$ \label{sampling_vertices_2}
\EndIf
\If {$z_2 =1$}
\State Add $(v,\cnt[v]$) to $V'$ and increment $\cnt[v]$ by $1$ \label{sampling_vertices_3}
\EndIf
\If{$|V'|>n^{1-3\beta/4}$}  \label{alg3_terminate_v'}
\State Terminate the execution
\EndIf
{\color{red}\State Sample $y_1,y_2 \sim \Bern(n^{-\beta/4})$ independently
\If {$y_1 =1$}
\State Increment $\estdeg[u]$ by $n^{\beta/4}$ 
\EndIf
\If {$y_2 =1$}
\State Increment $\estdeg[v]$ by $n^{\beta/4}$
\EndIf
\If {\# non-zero entries in $\estdeg> n^{1-\beta/8}$ }\label{Terminate_condition_est-deg}
\State Terminate the execution
\EndIf
}
\EndFor
\Statex
\For{every $v\in V$ such that $\cnt[v] > 0$}
\State  $\forall i\in [\cnt[v]]$, initialize $d(v,i)=0$
\EndFor
\State {\color{red} Initialize an $n$-length array $\storedeg$ (indexed by $v$) to be $0^n$}
\State {\color{red} Initialize $\hat{E} = \emptyset$ }
\Statex
\Statex \textbf{Second pass:}
\If{$m\le n^{1-\epsilon^{1/\epsilon}}$}
\State Store all the edges and compute the $\mdcut$ value.
\Else
\State Set a parameter $d=320/\epsilon^2$
\For{every edge $u\rightarrow v$ in the stream}
{\color{red} \If{$\cnt[u]>0$ and $\storedeg[u]< n^{2\beta/3}$}
\State Add $u\rightarrow v$ to $\hat{E}$ \label{add_edge_ehat}
\State Increment $\storedeg[u]$
\EndIf
\If{$\cnt[v]>0$ and $\storedeg[v]< n^{2\beta/3}$}
\State Add edge $u\rightarrow v$ to $\hat{E}$ if it was not previously added in \cref{add_edge_ehat} 
\State Increment $\storedeg[v]$
\EndIf}
\State Initialize $\mathsf{i}=1$
\For {$\mathsf{i}\le d$}
\If{$\cnt[u]>0$ and {\color{red} $\estdeg(u)>0$}}
\State {\color{red}  Sample $i_1\sim \Unif{[\estdeg(u)]}$} \label{alg:sample_i1}
\If {$i_1< \cnt[u]$}
\State Increment $d(u,i_1)$ by $1$ 
\EndIf
\EndIf
\If{$\cnt[v]>0$ and {\color{red} $\estdeg(v)>0$}}
\State {\color{red} Sample $i_2\sim \Unif{[\estdeg(v)]}$} \label{alg:sample_i2}
\If {$i_2< \cnt[v]$}
\State Increment $d(v,i_2)$ by $1$
\EndIf
\EndIf
\If{($\cnt[u]>0$, {\color{red} $\estdeg(u)>0$}, $i_1< \cnt[u]$) {\color{red} or} ($\cnt[v]>0$, {\color{red} $\estdeg(v)>0$}, $i_2< \cnt[v]$)}\label{alg3_store_alledges}
\State Add $(u,i_1)\rightarrow(v,i_2)$ to $E'$ 
\If{$|E'|>n^{1-\epsilon^{25/\epsilon}}$} \label{alg3_terminate_E'}
\State Terminate the execution
\EndIf
\EndIf
\State Increment $\mathsf{i}$ by $1$
\EndFor
\EndFor
\EndIf
\Statex
\Statex {\color{red} \textbf{Post-processing I: Resampling edges incident to low-degree vertices}:
\State Initialize $E''=E'$ \label{post_processing_start}
\For{every $u\in V$ such that $\cnt[u] > 0$ and $\storedeg[u] < n^{2\beta/3}$} \label{alg_reset_1}
\State Initialize $i=0$
\For{$i\in [\estdeg[u]]$} \label{alg_reset_2}
\For{every edge $(u,i)\rightarrow(v,j)\in E''$ or $(v,j)\rightarrow(u,i)\in E''$}
\State Delete the edge; decrement $d(u,i)$ (resp. $d(v,j)$) if $(u,i)\in V'$ (resp. $(v,j)\in V'$) \label{alg_reset_3}
\EndFor
\EndFor \label{alg_reset_5}
\EndFor \label{alg_reset_6}
\For{every $u\rightarrow v\in \hat{E}$}\label{post_resample_start}
\State Initialize $\mathsf{j}=1$
\For {$\mathsf{j}\le d$}
\If{$\cnt[u]>0$ and $\storedeg(u)< n^{2\beta/3}$}
\State Sample $j_1\sim \Unif{[\storedeg(u)]}$
\If{$j_1< \cnt[u]$}
\State Increment $d(u,j_1)$ by $1$
\EndIf
\If{$\cnt[v]>0$}
\If{$\storedeg(v)< n^{2\beta/3}$}
\State Sample $j_2\sim \Unif{[\storedeg(v)]}$
\Else
\State Sample $j_2\sim \Unif{[\estdeg(v)]}$
\EndIf
\If {$j_2< \cnt[v]$}
\State Increment $d(v,j_2)$ by $1$
\If{$j_1< \cnt[u]$}
\State Add $(u,j_1)\rightarrow(v,j_2)$ to $E''$
\EndIf
\EndIf
\EndIf
\ElsIf{$\cnt[v]>0$ and $\storedeg(v)< n^{2\beta/3}$}
\State Sample $j_2\sim \Unif{[\storedeg(v)]}$
\If{$j_2< \cnt[v]$}
\State Increment $d(v,j_2)$ by $1$
\EndIf
\If{$\cnt[u]>0$}
\State Sample $j_1\sim \Unif{[\estdeg(u)]}$
\If {$j_1< \cnt[u]$}
\State Increment $d(u,j_1)$ by $1$
\If{$j_2< \cnt[v]$}
\State Add $(u,j_1)\rightarrow(v,j_2)$ to $E''$
\EndIf
\EndIf
\EndIf
\EndIf
\State Increment $\mathsf{j}$ by $1$
\EndFor
\EndFor
} 
\State {\color{red} Delete every edge in $E''$ that is incident on a vertex not in $V'$} \label{post_processing_end}
\end{algorithmic}
\end{breakablealgorithm}

\paragraph*{Post-processing II (analogous to the post processing steps in \cref{Alg:bounded_avg_degree_threepass}):}
\begin{enumerate}

    \item For every $(u,i)\in V'$ such that $d(u,i) > 11 d=3520/\epsilon^2$, delete every edge incident on $(u,i)$ in ${\color{red} E''}$ and decrement $d(\cdot)$ each of its endpoints by $1$. Finally, set $d(u,i) = 0$. Let $\mathcal{G}(V',{\color{red} E''})$ be the directed graph on vertex set $V'$ and edge set ${\color{red} E''}$.
    \item  Let $\ell, c\in \mathbb{N}$ be set according to \cref{corollary_SSSV} for parameter $\epsilon''=\epsilon/8$. Let $D=11d$. Let $\mathcal{H}_{2m}$ be the hash family in \cref{corollary_SSSV} corresponding to directed graphs on $2m$ vertices. Sample $h\sim \Unif{\mathcal{H}_{2m}}$.
    \item For $T\in \DbAllTypesDeg{c}{\ell}{D}$, let $a_T$ denote the number of edges $e\in {\color{red} E''}$ such that every vertex $(v,i)\in V'$ which is at a distance of at most $\ell-1$ from at least one of the endpoints of $e$ satisfies the property that $\deg_{\mathcal{G}}(v,i) = d(v,i)$, and $\nbrtype{\mathcal{G}}{h|_{V'}}{\ell}{e} = T$.
    \item Consider the distribution $\mathcal{D}\in \Delta(\DbAllTypesDeg{c}{\ell}{D})$, defined as
    \[\mathcal{D}(T) = \frac{a_T \cdot n^{\beta |T|}}{\sum_{T'\in \DbAllTypesDeg{c}{\ell}{D}}a_{T'} \cdot n^{\beta|T'|}} \, , \] where $T\in\DbAllTypesDeg{c}{\ell}{D}$, and $|T|$ denotes the number of vertices in $T$.
    \item Let $\Local : \DbAllTypesDeg{c}{\ell}{D} \to [0,1]$ be the function in \cref{corollary_SSSV}. Output $\Exp_{T \sim \mathcal{D}}[\Local(T)]$.
\end{enumerate}

\begin{theorem}\label{thm: correctness of algorithm 2}
    For every $\epsilon\in (0,1/2)$, there exists $\mathsf{n}_0\in \mathbb{N}$ such that for all $n\ge \mathsf{n}_0$ and every directed graph $G$ on $n$ vertices and $m \le n^{1+\epsilon^{25/\epsilon}}$ edges, \cref{Alg:bounded_avg_degree_twopass} outputs a value $v$ satisfying
    \[
    \frac{\mdcut(G)}{2} - \epsilon \le v \le \mdcut(G) + \epsilon\, ,
    \] with probability at least $2/3$, using space $n^{1-\epsilon^{O(1/\epsilon)}}$.
\end{theorem} 

\begin{proof}
At a high level, the main change from \cref{Alg:bounded_avg_degree_threepass} is that for vertices randomly sampled from $\overline{V}$ in the first-pass, we don't use an additional pass like in \cref{Alg:bounded_avg_degree_threepass} to exactly compute the degrees of all their parent vertices. Instead, in the first pass, we approximately estimate the degrees of all high-degree vertices (of degree at least $n^{2\beta/3}$) with high probability. As shown in \cref{thm:trevisan_reduction}, for high-degree vertices, an approximate estimate of their degrees suffices for the reduction to work.
In the second pass, we handle the low degree vertices and store all the edges incident to parent vertices of degree less than $n^{2\beta/3}$ and resample their copies in $\overline{G}$ in the post-processing steps from \cref{post_processing_start} to \cref{post_processing_end}.
In addition, for technical reasons, $E'$ now stores \emph{all} edges incident on vertices in $V'$, not just the induced edges (see \cref{alg3_store_alledges}). This helps us maintain accurate degrees of the sampled vertices when we delete the edges sampled in the second pass that were incident to copies of low degree vertices and resample them later in post-processing I.
We now formally analyze the space usage and argue the correctness of \cref{Alg:bounded_avg_degree_twopass}.

\paragraph*{Space usage of \cref{Alg:bounded_avg_degree_twopass}.} Compared to \cref{Alg:bounded_avg_degree_threepass}, \cref{Alg:bounded_avg_degree_twopass} also stores $\estdeg$, $\storedeg$, and $\hat{E}$. It follows from the termination condition in \cref{Terminate_condition_est-deg} that $\estdeg$ has at most $n^{1-\beta/8}$ non-zero entries, and the value stored in each entry is at most $\text{poly}(n)$. Thus, the space required to store $\estdeg$ is at most $n^{1-\epsilon^{O(1/\epsilon)}}$. The termination condition in \cref{alg3_terminate_v'} ensures that $|V'|\le n^{1-3\beta/4}$. The number of non-zero entries in $\storedeg$ is at most $|V'|$. Finally, whenever an edge is added to $\hat{E}$, both its endpoints have copies in $V'$ and at least one of them currently has fewer than $n^{2\beta/3}$ edges incident on it in $\hat{E}$. Therefore, the number of edges stored in $\hat{E}$ is at most $|V'| n^{2\beta/3} \le n^{1-\beta/12}$.

\paragraph*{Correctness of \cref{Alg:bounded_avg_degree_twopass}.} For now, let us ignore the termination conditions in \cref{alg3_terminate_v',Terminate_condition_est-deg,alg3_terminate_E'}. We will later argue that these events occur with very low probability, and hence, do not alter the success probability of \cref{Alg:bounded_avg_degree_twopass} by much. We first prove the following claim.
\begin{claim}\label{clm:est-deg of high-degree vertices}
 Let $n$ be sufficiently large. With probability at least $1-\exp(n^{-\Omega(\epsilon^{O(1/\epsilon)})})$, for every $v\in V$ such that $\deg(v) \ge  n^{2\beta/3}$, we have \[(1-\epsilon/100)\deg(v) \le \mathsf{est}\text{-}\mathsf{deg}[v]  \le (1+\epsilon/100)\deg(v) \, .\]
\end{claim}
\begin{proof}
    Fix any vertex $v\in V$ whose degree is at least $n^{2\beta/3}$. For $i\in [\deg(v)]$, let $Y_i \sim \Bern(n^{-\beta/4})$ independently. Then, $\estdeg[v] = \sum_{i\in [\deg(v)]} Y_i \cdot n^{\beta/4}$. Therefore, $\Exp[\estdeg[v]] = \deg(v)$. By applying Hoeffding's inequality (\cref{thm:hoeffding's}), we conclude that with probability at least $1-\exp(n^{-\Omega(\epsilon^{O(1/\epsilon)})})$, \[(1-\epsilon/100)\deg(v) \le \mathsf{est}\text{-}\mathsf{deg}[v]  \le (1+\epsilon/100) \deg(v) \, .\] The claim follows from taking a union bound over at most $n$ vertices.
\end{proof}
Now let us condition on $\estdeg$ that satisfies the above property. We define an $n$-length array $\appdeg$ as follows. For every vertex $v\in V$, if $\deg(v) < n^{2\beta/3}$, then $\appdeg[v] := \deg(v)$, else $\appdeg[v] := \estdeg[v]$. Consider the bounded-degree graph $\overline{G}(\overline{V},\overline{E})$\footnote{We implicitly assume that \cref{Alg:trevisan} uses the same randomness as in \cref{Alg:bounded_avg_degree_twopass} to sample copies of the edges incident on the parent vertices of vertices in $V'$ and to sample the remaining edges, it uses independent random bits.} that is output by \cref{Alg:trevisan}, corresponding to the input graph $G(V,E)$, parameters $\epsilon'=\epsilon/2$ and $\zeta=2\beta/3$, and $\appdeg$. It follows from \cref{thm:trevisan_reduction} that with probability at least $5/6$,
\begin{itemize}
    \item $\mdcut(G)-\epsilon/2 \le \mdcut(\overline{G}) \le \mdcut(G) + \epsilon/2$,
    \item the maximum degree of any vertex in $\overline{G}$ is at most $110/(\epsilon')^2 = 3520/\epsilon^2$, and
    \item $\overline{G}$ has exactly $2m$ vertices and at least $240 m/\epsilon^2$ edges.
\end{itemize}
We will prove the following claim (analogous to \cref{claim_alg} in \cref{thm: correctness of algorithm 1}).
\begin{claim}\label{claim_alg2}
At the end of post-processing I, the following hold:
    \begin{enumerate}
    \item $V'$ has the same distribution as sampling every vertex from $\overline{V}$ independently with probability $n^{-\beta}$, and applying a relabeling on the vertices, and \label{item 1_claim_alg2}
    \item for $u\rightarrow v\in E$, if both $u$ and $v$ have copies in $V'$, then the copies of $u\rightarrow v$ in $\tilde{G}$ are sampled identically in \cref{Alg:trevisan,Alg:bounded_avg_degree_twopass}. If only one of $u$ or $v$ has a copy in $V'$, then the endpoints corresponding to the copies of that vertex in the copies of $u\rightarrow v$ in $\tilde{G}$ are sampled identically in \cref{Alg:trevisan,Alg:bounded_avg_degree_twopass}.\label{item 2_claim_alg2}
    \end{enumerate}
\end{claim}
\begin{proof}
The proof of \cref{item 1_claim_alg2} of \cref{claim_alg2} is analogous to the proof of \cref{item 1_claim_alg} of \cref{claim_alg}.
To prove \cref{item 2_claim_alg2}, we first consider the edges incident to high-degree vertices which have copies in $V'$.
Since $\appdeg$ and $\estdeg$ are the same for vertices of degree at least $n^{2\beta/3}$, the sampling in \cref{alg:sample_i1,alg:sample_i2} of \cref{Alg:bounded_avg_degree_twopass} is the same as the sampling in \cref{alg_samplei1_trev,alg_samplei2_trev} of \cref{Alg:trevisan} for edges whose both endpoints are of degree at least $n^{2\beta/3}$ in $G$. Similar, if only one of the endpoints of $u\rightarrow v$ has a copy in $V'$ and if the parent vertex of that copy is of degree at least $n^{2\beta/3}$, then the the endpoints corresponding to the copies of that vertex in the copies of $u\rightarrow v$ in $\tilde{G}$ are sampled identically in \cref{Alg:bounded_avg_degree_twopass,Alg:trevisan}.
Now consider the edges incident to low-degree vertices which have copies in $V'$. While these edges are \textbf{not} sampled with the correct probabilities in \cref{alg:sample_i1,alg:sample_i2}, this is fixed in the post processing steps from \cref{post_processing_start} to \cref{post_processing_end}.
For every $(u,i)\in V'$, observe that $\storedeg[u]<n^{2\beta/3}$ if and only if $\deg(u)<n^{2\beta/3}$. In addition, when $\deg(u)<n^{2\beta/3}$, $\storedeg[u] = \deg[u]$. In post-processing I, for every low-degree vertex that has a copy in $V'$, we delete all edges in $E''$ that are incident on \emph{any} of its copies, including those not contained in $V'$. All edges incident to low-degree vertices (which have a copy in $V'$) are stored in $\hat{E}$. Since $\storedeg[u] = \deg[u]$ for these vertices, we resample the copies of the edges incident on them identical to the sampling in \cref{Alg:trevisan}.
\end{proof}
Given the above claim, we consider $\tilde{G}$ where for $u\rightarrow v\in E$, if both $u$ and $v$ have copies in $V'$, then the copies of $u\rightarrow v$ are sampled according to \cref{Alg:bounded_avg_degree_twopass} (prior to post-processing II). If only one of $u$ or $v$ has a copy in $V'$, then the endpoints corresponding to the copies of that vertex in the copies of $u\rightarrow v$ are sampled according to \cref{Alg:bounded_avg_degree_twopass} and the other endpoints are sampled independently, according to \cref{Alg:trevisan}. The remaining edges are also sampled according to \cref{Alg:trevisan}.
Finally, like in \cref{Alg:trevisan}, we define $\overline{G}$ as $\tilde{G}$ with edges incident on vertices of degree larger than $11d$ deleted.
Observe that for $\overline{G}$ as defined above, $\mathcal{G}(V',E'')$ is the induced subgraph of $\overline{G}$ on $V'$.
Furthermore, at the end of post-processing II of \cref{Alg:bounded_avg_degree_twopass}, for $(v,i)\in V'$, $d(v,i)$ is degree of $(v,i)$ in $\overline{G}$. Now, the rest of the proof of correctness follows similarly as in the proof of \cref{thm: correctness of algorithm 1}. 

Finally, we argue that the termination conditions in \cref{alg3_terminate_v',Terminate_condition_est-deg,alg3_terminate_E'} happen with very low probability. Observe that $V'$ is sampled exactly in the same way as in \cref{Alg:bounded_avg_degree_threepass}. Therefore, the same proof goes through. To analyze the number of non-zero entries in $\estdeg$, for $i\in[\log n]$, consider any vertex $v\in V$ whose degree is larger than $2^i$ and at most $2^{i+1}$. Let $Z^i_v$ be the indicator variable for the event that $\estdeg[v]\ne 0$ and let $Z^i=\sum_{v:2^i< \deg(v)\le 2^{i+1}} Z^i_v$. Observe that $\Exp[Z^i_v]\le 2^{i+1}n^{-\beta/4}$ and by the linearity of expectation, \[\Exp[Z^i] \le \frac{2m}{2^i} \cdot 2^{i+1}n^{-\beta/4} = 4 m n^{-\beta/4} \le n^{1-\beta/6} \, .\] Since the events are independent, applying Chernoff bound (\cref{prelim:chernoff}), we conclude that with probability at least $1-\exp(n^{-\Omega(\epsilon^{O(1/\epsilon)})})$, $Z^i \le (\log n) n^{1-\beta/6}$. Taking a union bound over all $i\in [\log n]$, we get that with high probability, $\sum_{i\in [\log n] } Z^i \le (\log n)^2 n^{1-\beta/6} < n^{1-\beta/8}$.
We now analyze the termination condition on $E'$ in \cref{alg3_terminate_E'}. 
Let $E'_\ell$ denote the set of edges in $E'$ which are incident to vertices in $V'$ whose parent vertices have degree less than $n^{2\beta/3}$.
Let $E'_h$ denote the set of edges in $E'$ which are incident to vertices in $V'$ whose parent vertices have degree at least $n^{2\beta/3}$.
Observe that $E' = E'_\ell \cup E'_h$.
We previously argued (in \cref{sec:algorithms}) that with high probability, $|V'|\le n^{1-3\beta/4}$. Therefore, with high probability, $|E'_\ell|\le n^{1-\beta/12}$.
It follows from \cref{clm:est-deg of high-degree vertices} that with high probability, for every $v\in V$ such that $\deg(v) \ge  n^{2\beta/3}$, we have \[(1-\epsilon/100)\deg(v) \le \mathsf{est}\text{-}\mathsf{deg}[v]  \le (1+\epsilon/100)\deg(v) \, .\]
Therefore, conditioned on this event, $\Exp[|E'_h|]\le d n^{1+\delta} (2n^{-\beta}(1-\epsilon/100)^{-1}) \le n^{1-3\beta/4}$, for sufficiently large $n$. Applying Markov's inequality, we conclude that with high probability, $|E'_h|\le n^{1-2\beta/3}$. Thus, the termination condition on $E'$ in \cref{alg3_terminate_E'} occurs with very low probability, and this concludes the proof of correctness of \cref{Alg:bounded_avg_degree_twopass}.

\end{proof}
Below, we describe our algorithm for $\mdcut$ on arbitrary graphs.
\begin{algorithm}[H]
\caption{Two-pass streaming algorithm for $\mdcut$}\label{Alg:meta}
\begin{algorithmic}[1]
\Statex \textbf{Parameters:} $n \in \mathbb{N}$, $\epsilon >0$
\Statex \textbf{Input:} A stream $\vecsigma$ of edges of a directed graph $G(V,E)$ on $n$ vertices
\Statex

\Statex \textbf{First pass:}
\State Maintain a counter for the number of edges $m$ and execute the first passes of \cref{Alg:bounded_avg_degree_twopass} and the algorithm in  \cite{DBLP:conf/icalp/BhaskaraDV18}, in parallel
\Statex
\Statex \textbf{Second pass:}
\If {$m\le n^{1+\epsilon^{4/\epsilon}}$}
\State Continue executing the second pass of \cref{Alg:bounded_avg_degree_twopass} and abandon the execution of \cite{DBLP:conf/icalp/BhaskaraDV18}
\Else 
\State Abandon the execution of \cref{Alg:bounded_avg_degree_twopass} and execute the second pass of \cite{DBLP:conf/icalp/BhaskaraDV18}
\EndIf
\end{algorithmic}
\end{algorithm}

Finally, we conclude with a proof of our main theorem.
\begin{proof}[Proof of \cref{thm:main}]
The proof of correctness and the space usage of \cref{Alg:meta} immediately follows from \cref{thm: correctness of algorithm 2} and \cref{thm:aditya}.  
\end{proof}
\ifnum\doubleblind=0
\section*{Acknowledgment}
The author is supported in part by NSF award CCF 2348475. Part of the work was conducted when the author was visiting the Simons Institute for the Theory of Computing as a research fellow in the Sublinear Algorithms program. The author thanks Raghuvansh Saxena, Noah Singer, and Madhu Sudan for collaborating on earlier investigations into this problem.  The author thanks Aditya Bhaskara for discussions related to \cite{DBLP:conf/icalp/BhaskaraDV18}.
\fi
\printbibliography
\end{document}